\documentclass[11pt]{article}
\usepackage{a4wide}
\usepackage[usenames,dvipsnames]{color}
\usepackage{amssymb,amsmath,graphics,color,enumerate,eucal}
\usepackage{tikz}
\usepackage{xspace}
\usetikzlibrary{decorations.pathmorphing,patterns,calc}
\usepackage{todonotes}
\usepackage[utf8x,utf8]{inputenc} 
\usepackage[T1]{fontenc}
\usepackage[vlined,linesnumbered,boxruled]{algorithm2e}

\usepackage{epsfig}
\usepackage{subfig}
\usepackage{amsmath,amsfonts,amssymb,epsfig,color,amsthm}

\newtheorem{theorem}{Theorem}[]
\newtheorem{lemma}[theorem]{Lemma}

\newcommand{\mycase}[1]{\noindent{\bf CASE #1\ }}

\newcommand{\ignore}[1]{}%

\newcommand{\ProblemFormat}[1]{{\sc #1}}
\newcommand{\ProblemName}[1]{\ProblemFormat{#1}\xspace}

\newcommand{\probPlMDS}{\ProblemName{Planar Dominating Set}}
\newcommand{\probEDS}{\ProblemName{Edge Dominating Set}}

\newcommand{\probCVC}{\ProblemName{Connected Vertex Cover}}
\newcommand{\probTP}{\ProblemName{Triangle Packing}}
\newcommand{\probFVS}{\ProblemName{Feedback Vertex Set}}
\newcommand{\probPlFVS}{\ProblemName{Planar Feedback Vertex Set}}

\newcommand{\NP}{{\ensuremath{\rm{NP}}}}
\newcommand{\coNP}{{\ensuremath{\rm{coNP}}}}
\DeclareMathOperator{\poly}{poly}

\newcommand{\heading}[1]{\medskip\noindent{\bf #1.\ }}%

\newcounter{rulecnt}
\newcommand{\rrule}[2]{\medskip\noindent{\bf \refstepcounter{rulecnt}\label{#1}Rule~\ref{#1}}  #2\ }%

\newcommand{\Ff}{{\ensuremath{\mathcal{F}}}}

\newcommand{\Ss}{{\ensuremath{\mathcal{S}}}}

\newcommand{\scalefactor}{0.8}

\begin{document}
\title{A $13k$-kernel for Planar Feedback Vertex Set\\ via Region Decomposition\thanks{A preliminary version (with a slightly weaker result) was presented at IPEC 2014, Wrocław, Poland. Work partially supported by the ANR Grant EGOS (2012-2015) 12 JS02 002 01 (MB) and by the National Science Centre of Poland, grant number UMO-2013/09/B/ST6/03136 (\L{}K).} 
}

\date{}

\author{Marthe Bonamy\thanks{LIRMM, France} \and \L ukasz Kowalik\thanks{University of Warsaw, Poland}}

\maketitle

\begin{abstract}
We show a kernel of at most $13k$ vertices for the \probFVS problem restricted to planar graphs, i.e., a polynomial-time algorithm that transforms an input instance $(G,k)$ to an equivalent instance with at most $13k$ vertices.
To this end we introduce a few new reduction rules.
However, our main contribution is an application of the region decomposition technique in the analysis of the kernel size.
We show that our analysis is tight, up to a constant additive term.
\end{abstract}

\section{Introduction}
A {\em feedback vertex set} in a graph $G=(V,E)$ is a set of vertices $S\subseteq V$ such that $G-S$ is a forest.
In the \probFVS problem, given a graph $G$ and integer $k$ one has to decide whether $G$ has a feedback vertex set of size $k$.
This is one of the fundamental NP-complete problems, in particular it is among the 21 problems considered by Karp~\cite{Kar72}.
It has applications e.g.\ in operating systems (see~\cite{sgg}), VLSI design, synchronous systems and artificial intelligence (see~\cite{festa2009feedback}).

In this paper we study kernelization algorithms, i.e., polynomial-time algorithms which, for an input instance $(G,k)$ either conclude that $G$ has no feedback vertex set of size $k$ or return an equivalent instance $(G',k')$, called {\em kernel}. 
In this paper, by the size of the kernel we mean the number of vertices of $G'$.
Burrage et al.~\cite{Burrage06} showed that \probFVS has a kernel of size $O(k^{11})$, which was next improved to $O(k^3)$ by Bodlaender~\cite{Bodlaender07} and to $4k^2$ by Thomass\'e~\cite{fvs:quadratic-kernel}.
Actually, as argued by Dell and van Melkebeek~\cite{DellM14} the kernel of Thomass\'e can be easily tuned to have the number of {\em edges} bounded by $O(k^2)$. This cannot be improved to $O(k^{2-\epsilon})$ for any $\epsilon>0$, unless $\coNP \subseteq \NP/\poly$~\cite{DellM14}.

In this paper we study \probPlFVS problem, i.e., \probFVS restricted to planar graphs.
Planar versions of NP-complete graph problems often enjoy kernels with $O(k)$ vertices.
Since an $n$-vertex planar graph has $O(n)$ edges, this implies they have $O(k)$ edges, and hence are called {\em linear kernels}.
The first nontrivial result of that kind was presented in the seminal work of Alber, Fellows and Niedermeier~\cite{afn:planar-domset} who showed a kernel of size $335k$ for \probPlMDS.
One of the key concepts of their paper was the region decomposition technique in the analysis of the kernel size.
Roughly, in this method the reduced plane instance is decomposed into $O(k)$ {\em regions} (i.e. subsets of the plane) such that every region contains $O(1)$ vertices of the graph. 
It was next applied by Guo and Niedermeier to a few more graph problems~\cite{GuoN07}.
In fact it turns out that for a number of problems on planar graphs, including \probPlMDS and \probPlFVS, one can get a kernel of size $O(k)$ by general method of protrusion decomposition~\cite{fomin:bidim-kernels}.
However, in this general algorithm the constants hidden in the $O$ notation are very large, and researchers keep working on problem-specific linear kernels with the constants as small as possible~\cite{cfkx:duality-and-vertex,LuoWFGC13,wygc13,KowalikPS13,Kowalik12}.

In the case of \probPlFVS, Bodlaender and Penninkx~\cite{bp08} gave an algorithm which outputs a kernel of size at most $112k$.
This was next improved by Abu-Khzam and Khuzam~\cite{ak12} to $97k$.
Very recently, and independently of our work, Xiao~\cite{Xiao14} has presented an improved kernel of $29k$ vertices.
However, neither of these papers uses the region decomposition.
Indeed, it seems non-obvious how the regions of the region decomposition can be defined for \probPlFVS. 
Instead, the authors of the previous works cleverly apply simple bounds on the number of edges in general and bipartite planar graphs.
Moreover, for certain problems these methods turned out to give better results and simpler proofs than those based on region decomposition, see e.g., the work of Wang, Yang, Guo and Chen~\cite{wygc13} on \probCVC, \probEDS, and \probTP in planar graphs improving previous results of Guo and Niedermeier~\cite{GuoN07}.

Somewhat surprisingly, in this work we show that region decomposition can be successfully applied to \probPlFVS, and moreover it gives much tighter bounds than the previous methods. 
Furthermore, we add a few new reduction rules to improve the bound even further, to $13k$.
More precisely, we show the following result.

\begin{theorem}\label{thm:FVS}
There is an algorithm that, given an instance $(G,k)$ of \probPlFVS, either reports that $G$ has no feedback vertex set of size $k$ or produces an equivalent instance with at most $13k-24$ vertices. The algorithm runs in expected $O(n)$ time, where $n$ is the number of vertices of $G$.
\end{theorem}

We use the region decomposition approach in a slightly relaxed way: the regions are the faces of a $k$-vertex plane graph and the number of vertices of the reduced graph in each region is linear in the length of the corresponding face.
We show that this gives a {\em tight bound}, i.e., we present a family of graphs which can be returned by our algorithm and have $13k-O(1)$ vertices.

\heading{Organization of the paper}
In Section~\ref{sec:16k} we present a kernelization algorithm which is obtained from the algorithms in~\cite{bp08,ak12} by generalizing a few reduction rules, and adding some completely new rules. 
In Section~\ref{sec:analysis} we present an analysis of the size of the kernel obtained by our algorithm.
In the analysis we assume that in the reduced graph, for every induced path with $\ell$ internal vertices, the internal vertices have at least three neighbors outside the path.
Based on this, we get the bound of $(2\ell+3)k- (4\ell+4)$ for the number of vertices in the kernel.
In Section~\ref{sec:16k} we present reduction rules which guarantee that in the kernel $\ell\le 6$, resulting in the kernel size bound of $15k-28$.
To get the claimed bound of $13k-24$ vertices in Section~\ref{sec:5path} we present a complex set of reduction rules, which allow us to conclude that $\ell\le 5$.
In Section~\ref{sec:time} we discuss the running time of the algorithm.
Finally, in Section~\ref{sec:conclude} we discuss possibilities of further research.

\heading{Notation}
In this paper we deal with multigraphs, though for simplicity we refer to them as graphs.
(Even if the input graph is simple, our algorithm may introduce multiple edges.)
By the {\em degree} of a vertex $x$ in a multigraph $G$, denoted by $\deg_G(x)$, we mean the number of edges incident to $x$ in $G$.
By $N_G(x)$, or shortly $N(x)$, we denote the set of neighbors of $x$, while $N[x]=N(x)\cup\{x\}$ is the closed neighborhood of $x$.
Note that in a multigraph $|N_G(x)|\le \deg_G(x)$, but the equality does not need to hold.
The neighborhood of a set of vertices $S$ is defined as $N(S)=(\bigcup_{v\in S}N(v))\setminus S$, while the closed neighborhood of $S$ is $N[S]=(\bigcup_{v\in S}N(v)) \cup S$.
For a face $f$ in a plane graph, a {\em facial walk} of $f$ is the shortest closed walk induced by all edges incident with $f$.
The {\em length} of $f$, denoted by $d(f)$ is the length of its facial walk.

\section{Our kernelization algorithm}
\label{sec:16k}

In this section we describe our algorithm which outputs a kernel for \probPlFVS. 
The algorithm exhaustively applies reduction rules. 
Each reduction rule is a subroutine which finds in polynomial time a certain structure in the graph and replaces it by another structure, so that the resulting instance is equivalent to the original one. 
More precisely, we say that a reduction rule for parameterized graph problem $P$ is {\em correct} when for every instance $(G,k)$ of $P$ it returns an instance $(G',k')$ such that:
\begin{enumerate}[a)]
\item $(G',k')$ is an instance of $P$,
\item $(G,k)$ is a yes-instance of $P$ iff $(G',k')$ is a yes-instance of $P$, and
\item $k'\le k$.
\end{enumerate}

Below we state the rules we use. 
The rules are applied in the given order, i.e., in each rule we assume that the earlier rules do not apply.
We begin with some rules used in the previous works~\cite{ak12,bp08}.

\captionsetup[subfloat]{labelformat=empty}
\captionsetup{font=small,format=plain,labelfont=bf}
\begin{figure}[t]
\centering
\subfloat[][Rule~\ref{r:deg1}]{
\centering
\begin{tikzpicture}[scale=\scalefactor]
\tikzstyle{whitenode}=[draw,circle,fill=white,minimum size=5pt,inner sep=0pt]
\tikzstyle{blacknode}=[draw,circle,fill=black,minimum size=5pt,inner sep=0pt]
\tikzstyle{texte}=[circle,minimum size=5pt,inner sep=0pt]
\draw (0,0) node[whitenode] (u) {}
-- ++(0:1cm) node[blacknode] (y1) {}; 
\draw (0.5,-0.5) node[texte] (k) {$k$};

\draw (1.5,-0.2) node[texte] (k) {\Large $\leadsto$};

\draw (2,0) node[whitenode] (u) {};
\draw (2,-0.5) node[texte] (k) {$k$};
\end{tikzpicture}
\label{fig:r:deg1}
}
\qquad
\subfloat[][Rule~\ref{r:deg2}]{
\centering
\begin{tikzpicture}[scale=\scalefactor]
\tikzstyle{whitenode}=[draw,circle,fill=white,minimum size=5pt,inner sep=0pt]
\tikzstyle{blacknode}=[draw,circle,fill=black,minimum size=5pt,inner sep=0pt]
\tikzstyle{texte}=[circle,minimum size=5pt,inner sep=0pt]
\tikzstyle{innerWhite} = [semithick, white,line width=0.5pt]
\draw (0,0) node[whitenode] [label=90:$w$] (u) {}
-- ++(0:1cm) node[blacknode] [label=90:$u$] (y1) {}
-- ++(0:1cm) node[whitenode] [label=90:$v$] (y2) {}; 
\draw (1,-0.5) node[texte] (k) {$k$};

\draw (2.5,-0.2) node[texte] (k) {\Large $\leadsto$};

\draw (3,0) node[whitenode] [label=90:$w$] (u) {};
\draw (4,0) node[whitenode] [label=90:$v$] (v) {};
\draw (u) edge node {} (v); 
\draw (3.5,-0.5) node[texte] (k) {$k$};
\end{tikzpicture}
\label{fig:r:deg2}
}
\qquad
\subfloat[][Rule~\ref{r:deg3-double}]{
\centering
\begin{tikzpicture}[scale=\scalefactor]
\tikzstyle{whitenode}=[draw,circle,fill=white,minimum size=5pt,inner sep=0pt]
\tikzstyle{blacknode}=[draw,circle,fill=black,minimum size=5pt,inner sep=0pt]
\tikzstyle{texte}=[circle,minimum size=5pt,inner sep=0pt]
\tikzstyle{innerWhite} = [semithick, white,line width=0.5pt]
\draw (0,0) node[whitenode] [label=90:$w$] (u) {}
-- ++(0:1cm) node[blacknode] [label=90:$u$] (y1) {};
\draw (2,0) node[whitenode] [label=90:$v$] (y2) {}; 
\draw (y1) edge [bend left] node {} (y2);
\draw (y1) edge [bend right] node {} (y2);
\draw (1,-0.5) node[texte] (k) {$k$};

\draw (2.5,-0.2) node[texte] (k) {\Large $\leadsto$};

\draw (3.5,0) node[whitenode] [label=90:$w$] (u) {};
\draw (3.5,-0.5) node[texte] (k) {$k-1$};
\end{tikzpicture}
\label{fig:r:deg3-double}
}
\\
\subfloat[][Rule~\ref{r:triple}]{
\centering
\begin{tikzpicture}[scale=\scalefactor]
\tikzstyle{whitenode}=[draw,circle,fill=white,minimum size=5pt,inner sep=0pt]
\tikzstyle{blacknode}=[draw,circle,fill=black,minimum size=5pt,inner sep=0pt]
\tikzstyle{texte}=[circle,minimum size=5pt,inner sep=0pt]
\tikzstyle{innerWhite} = [semithick, white,line width=0.5pt]
\draw (0,0) node[whitenode] (u) {};
\draw (1,0) node[whitenode] (y2) {}; 
\draw (u) edge [bend left=60] node {} (y2);
\draw (u) edge [bend left=40] node {} (y2);
\draw (u) edge [bend right=40] node {} (y2);

\draw (0.5,0.15) edge [thick,densely dotted] node {} (0.5,-0.15);

\draw (0.5,-0.5) node[texte] (k) {$k$};

\draw (1.5,-0.2) node[texte] (k) {\Large $\leadsto$};

\draw (2,0) node[whitenode](u) {};
\draw (3,0) node[whitenode] (y2) {}; 
\draw (u) edge [bend left] node {} (y2);
\draw (u) edge [bend right] node {} (y2);
\draw (2.5,-0.5) node[texte] (k) {$k$};
\end{tikzpicture}
\label{fig:r:triple}
}
\qquad
\subfloat[][Rule~\ref{r:3deg3}]{
\centering
\begin{tikzpicture}[scale=\scalefactor]
\tikzstyle{whitenode}=[draw,circle,fill=white,minimum size=5pt,inner sep=0pt]
\tikzstyle{blacknode}=[draw,circle,fill=black,minimum size=5pt,inner sep=0pt]
\tikzstyle{texte}=[circle,minimum size=5pt,inner sep=0pt]
\tikzstyle{innerWhite} = [semithick, white,line width=0.5pt]
\draw (0,0) node[blacknode] (a) {};
\draw (0,-0.5) node[blacknode] (b) {}; 
\draw (0,-1) node[blacknode] (c) {}; 

\draw (1,0) node[whitenode] (u) {}; 
\draw (1,-1) node[whitenode] (v) {}; 

\draw[densely dotted] (a)
-- ++(180:0.5cm);
\draw[densely dotted] (b)
-- ++(180:0.5cm);
\draw[densely dotted] (c)
-- ++(180:0.5cm);

\draw (a) edge [bend left=10] node {} (u);
\draw (a) edge [bend right=10,densely dotted] node {} (u);
\draw (a) edge [bend left=10] node {} (v);
\draw (a) edge [bend right=10,densely dotted] node {} (v);
\draw (b) edge [bend left=10] node {} (u);
\draw (b) edge [bend right=10,densely dotted] node {} (u);
\draw (b) edge [bend left=10] node {} (v);
\draw (b) edge [bend right=10,densely dotted] node {} (v);
\draw (c) edge [bend left=10] node {} (u);
\draw (c) edge [bend right=10,densely dotted] node {} (u);
\draw (c) edge [bend left=10] node {} (v);
\draw (c) edge [bend right=10,densely dotted] node {} (v);
\draw (0.5,-1.5) node[texte] (k) {$k$};

\draw (a) node[above] {$a$};
\draw (b) node[left] {$b$};
\draw (c) node[below=1mm] {$c$};
\draw (u) node[right] {$u$};
\draw (v) node[right] {$v$};

\draw (1.9,-0.6) node[texte] (k) {\Large $\leadsto$};

\draw (2.8,-1.5) node[texte] (k) {$k-2$};
\end{tikzpicture}
\label{fig:r:3deg3}
}
\qquad
\subfloat[][Rule~\ref{r:gamma}]{
\centering
\begin{tikzpicture}[scale=\scalefactor]
\tikzstyle{whitenode}=[draw,circle,fill=white,minimum size=5pt,inner sep=0pt]
\tikzstyle{blacknode}=[draw,circle,fill=black,minimum size=5pt,inner sep=0pt]
\tikzstyle{texte}=[circle,minimum size=5pt,inner sep=0pt]
\tikzstyle{innerWhite} = [semithick, white,line width=0.5pt]
\draw (0,0) node[blacknode] (u) {};
\draw (0,-1) node[whitenode] (v) {}; 
\draw (u)
--++(-30:1cm) node[whitenode] (w) {}; 
\draw (u)
--++(-150:1cm) node[whitenode] (x) {}; 

\draw (u) node [above=1mm] {$u$};
\draw (v) node [below=1mm] {$v$};
\draw (w) node [below=1mm] {$w$};
\draw (x) node [below=1mm] {$x$};

\draw (u) edge [bend left=10] node {} (v);
\draw (u) edge [bend right=10,densely dotted] node {} (v);
\draw (w) edge [bend left=10] node {} (v);
\draw (w) edge [bend right=10,densely dotted] node {} (v);
\draw (x) edge [bend left=10] node {} (v);
\draw (x) edge [bend right=10,densely dotted] node {} (v);
\draw (0,-1.8) node[texte] (k) {$k$};

\draw (1.4,-0.6) node[texte] (k) {\Large $\leadsto$};

\draw (2.8,-1) node[whitenode] (v) {}; 
\draw (2.8,0)
++(-30:1cm) node[whitenode] (w) {}; 
\draw (2.8,0)
++(-150:1cm) node[whitenode] (x) {}; 

\draw (w) edge [bend right=10] node {} (x);
\draw (w) edge [bend left=10,densely dotted] node {} (x);
\draw (w) edge [bend left=10] node {} (v);
\draw (w) edge [bend right=10] node {} (v);
\draw (x) edge [bend left=10] node {} (v);
\draw (x) edge [bend right=10] node {} (v);
\draw (2.8,-1.8) node[texte] (k) {$k$};
\draw (v) node [below=1mm] {$v$};
\draw (w) node [below=1mm] {$w$};
\draw (x) node [below=1mm] {$x$};
\end{tikzpicture}
\label{fig:r:gamma}
}
\caption{Reduction rules~\ref{r:loop}--\ref{r:gamma}. Dashed edges are optional. We draw in black the vertices whose incident edges are all already drawn (as solid or dashed edges), in white the vertices which might be incident to other edges. Regardless of their color, vertices in the figures may not coincide.
}
\label{fig:rules}
\vspace{-5mm}
\end{figure}
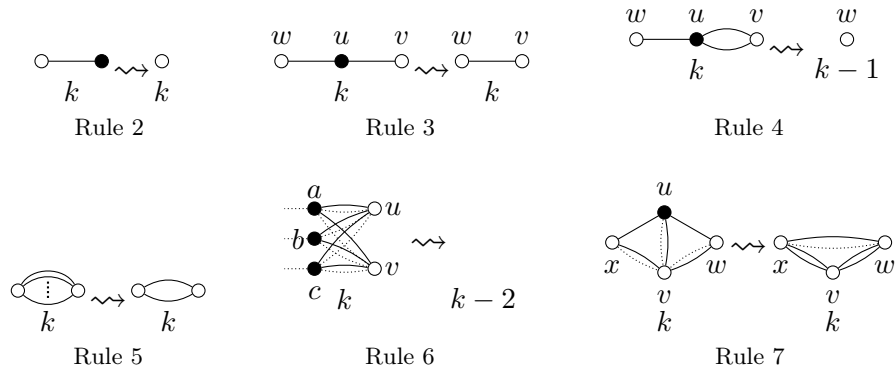

\rrule{r:loop}{If there is a loop at a vertex $v$, remove $v$ and decrease $k$ by one.}

\rrule{r:deg1}{Delete vertices of degree at most one.}

\rrule{r:deg2}{If a vertex $u$ is of degree two, with incident edges $uv$ and $uw$, then delete $u$ and add the edge $vw$. (Note that if $v=w$ then a loop is added.)}

\rrule{r:deg3-double}{If a vertex $u$ has exactly two neighbors $v$ and $w$, edge $uv$ is double, and edge $uw$ is simple, then delete $v$ and $u$ and decrease $k$ by one.}

\rrule{r:triple}{If there are at least three edges between a pair of vertices, remove all but two of the edges.}


\rrule{r:3deg3}{Assume that there are five vertices $a,b,c,v,w$ such that 
1) both $v$ and $w$ are neighbors of each of $a,b,c$ and 
2) each vertex $x\in\{a,b,c\}$ is incident with at most one edge $xy$ such that $y\not\in\{v,w\}$. 
Then remove all the five vertices and decrease $k$ by two.}

\medskip

The correctness of the above reduction rules was proven in~\cite{ak12}. 
(In~\cite{ak12}, Rule~\ref{r:3deg3} is formulated in a slightly less general way which forbids multiplicity of some edges, but the correctness proof stays the same.)
Now we introduce a few new rules.

\rrule{r:gamma}{If a vertex $u$ has exactly three neighbors $v$, $w$ and $x$, $v$ is also adjacent to $w$ and $x$, and both edges $uw$ and $ux$ are simple, then contract $uv$ and add an edge $wx$ (increasing its multiplicity if it already exists). If edge $uv$ was not simple, add a loop at $v$.}

\begin{lemma}\label{cl:r:gamma}
Rule~\ref{r:gamma} is correct.
\end{lemma}
\begin{proof}
Let $G'$ be the graph obtained from a graph $G$ by a single application of Rule~\ref{r:gamma}.  
Let $S$ be a feedback vertex set of size $k$ in $G'$. We claim $S$ is a feedback vertex set in $G$ too. 
Assume for a contradiction that there is a cycle $C$ in $G-S$.
Then $u\in V(C)$, for otherwise $C\subseteq G'$.
If $v\in S$ then $\{wu,ux\}\subseteq C$ and $C-\{wu,ux\}+\{wx\}$ is a cycle in $G'$, a contradiction.
If $v\not\in S$, then $w,x\in S$ and hence $v$ is the only neighbor of $u$ in $G-S$, so $C$ is the 2-cycle $uvu$. But then $G'-S$ contains a loop at $v$, a contradiction.

Let $S$ be a feedback vertex set of size $k$ in $G$. 
If $|\{u,v\}\cap S|=2$, then $S\setminus\{u\}\cup\{w\}$ is a feedback vertex set of size $k$ in $G'$.
Assume $|\{u,v\}\cap S|=1$. Then we can assume $v\in S$ for otherwise we replace $S$ by $S\setminus\{u\}\cup\{v\}$, which is also a feedback vertex set in $G$. If there is a cycle $C$ in $G'-S$ , then $wx\in E(C)$, for otherwise $C\subseteq G-S$.
But then $C-\{wx\}+\{wu,ux\}$ is a cycle in $G$, a contradiction.
Finally, if $|\{u,v\}\cap S|=0$ then both $w$ and $x$ are in $S$, so $S$ is also a feedback vertex set in $G'$.
\end{proof}

The graph modification in Rule~\ref{r:gamma} is an example of a {\em gadget replacement}, i.e., a subgraph of $G$ is replaced by another subgraph in such a way that the answer to the {\sc Feedback Vertex Set} problem does not change. 
We will use many rules of this kind, and their correctness proofs all use similar arguments.
In order to make our proofs more compact, we define gadget replacement formally below, and prove a technical lemma (Lemma~\ref{lem:fvs-transfer} below) which will be used in many rule correctness proofs.

{\em Gadget replacement} in graph $A$ is a triple $(X,Y,E_I)$, where $X\subseteq V(A)$, $Y$ is a set of vertices disjoint with $V(A)$, and $E_I$ is a set of edges with both endpoints in $N_A(X)\cup Y$.
The {\em result} of gadget replacement is a new graph $B$, obtained from $A$ by deleting $X$ and $E(G[N[X]])$ and inserting $Y$ and $E_I$.

For an example, in Rule~\ref{r:gamma}, $G'$ is a result of gadget replacement $(\{u\},\emptyset,\{xv,wv,xw\})$.
Note that if $(X,Y,E_I)$ is a gadget replacement in $A$ that results in $B$ then $(Y,X,\{uv\in E(A)\ :\ u,v\in N_A[X]\})$ is a gadget replacement in $B$ and its result is $A$.

\begin{lemma}
\label{lem:fvs-transfer}
 Let $(X,Y,E_I)$ be a gadget replacement in graph $A$, and let $B$ be its result.
 Let $Q_A=A[N_A[X]]$ and $Q_B=B[Y\cup N_A(X)]$.
 Let $S_A$ be a feedback vertex set in $A$.
 Let $S_B$ be a subset of vertices of $B$ such that $S_A\setminus V(Q_A) = S_B \setminus V(Q_B)$ and $Q_B-S_B$ is a forest.
 Finally, assume that for every pair $u,v\in N_A(X)$ if there is a $(u,v)$-path in $Q_B-S_B$ then there is a $(u,v)$-path in $Q_A-S_A$.
 Then $S_B$ is a feedback vertex set of $B$.
\end{lemma}

\begin{proof}
 Assume for a contradiction that there is a cycle $C$ in $B-S_B$.
 Since $Q_B - S_B$ is a forest, $C$ has at least one vertex outside $Q_B$.
 Assume $C$ has all vertices outside $Q_B$. 
 But then $C\subseteq A$ and since $S_A\setminus V(Q_A) = S_B \setminus V(Q_B)$ we also have $C\subseteq A-S_A$, so $S_A$ is not a feedback vertex set of $A$, a contradiction.
 
 Hence we know that $C$ has vertices both inside and outside $Q_B$.
 It follows that $C$ can be divided into subpaths of two kinds: subpaths in $Q_B-S_B$ and in $(B-E(Q_B))-S_B$.
 Every such subpath $P_B$ in $Q_B-S_B$ is of the form $v_1,\ldots,v_{t}$, where $v_1,v_t \in N_A(X)$ and $v_2,\ldots,v_{t-1} \in V(Q_B)-S_B$.
 Hence, by the assumption of the lemma, there is a $(v_1,v_t)$-path $P_A$ in $Q_A-S_A$.
 In particular, $v_1,v_t\not\in S_A$ (we will use this observation later).
 Let $C'$ be the closed walk obtained from $C$ by replacing every maximal subpath in $Q_B-S_B$ by a path in $Q_A-S_A$. 
 Now consider a maximal subpath $P_B'$ of $C$ in $(B-E(Q_B))-S_B$. 
 It is of the form $v_1,\ldots,v_{t}$, where $v_1,v_t \in N_A(X)$ and $v_2,\ldots,v_{t-1} \in (V(B)\setminus V(Q_B))\setminus S_B$.
 Since $v_1$ and $v_t$ are also endpoints of paths in $Q_B-S_B$, it holds that $v_1,v_t\not\in S_A$ as argued above.
 Since $S_A\setminus V(Q_A) = S_B \setminus V(Q_B)$, it holds that $v_2,\ldots,v_{t-1} \in V(A)\setminus S_A$.
 Hence $P_B'\subseteq A-S_A$. 
 It follows that $C'$ is a closed walk in $A-S_A$, hence there is a cycle in $A-S_A$, a contradiction.
\end{proof}

When applying Lemma~\ref{lem:fvs-transfer}, we will examine reachability relations in $Q_A-S_A$ and in $Q_B-S_B$, in order to check whether the last assumption holds. It will be convenient to introduce the following notation.
For a graph $H$ and set of vertices $S$, let $R_{H,S}$ be the reachability relation in $H$ truncated to $S$, i.e., $(a,b)\in R_{H,S}$ iff $a,b\in S$ and there is an $(a,b)$-path in $H$. 
The set $S$ does not need to be a subset of $V(H)$; for every vertex $a\in S\setminus V(H)$, $\{a\}$ forms an equivalence class of $R_{H,S}$.

\begin{figure}[t]
\centering
\begin{tikzpicture}[scale=\scalefactor]
\tikzstyle{whitenode}=[draw,circle,fill=white,minimum size=5pt,inner sep=0pt]
\tikzstyle{blacknode}=[draw,circle,fill=black,minimum size=5pt,inner sep=0pt]
\tikzstyle{texte}=[circle,minimum size=5pt,inner sep=0pt]
\tikzstyle{innerWhite} = [semithick, white,line width=0.5pt]
\draw (0,0) node[blacknode] [label=-90:$v_1$] (v1) {}
-- ++(180:0.5cm) node[blacknode] [label=-90:$v_2$] (v2) {}
 ++(180:0.5cm) node[whitenode] [label=-90:$v_3$] (v3) {};

\draw (v2) edge node {} (v3);

\draw (1,0) node[blacknode] [label=-90:$u_1$] (u1) {}
 ++(0:0.5cm) node[whitenode] [label=-90:$u_2$] (u2) {};
\draw (u1) edge node {} (u2);

\draw (0.5,1) node[whitenode] [label=right:$w_1$] (w1) {}; 
\draw (0.5,-1) node[whitenode] [label=right:$w_2$] (w2) {}; 

\draw (w1) edge node {} (u1);
\draw (w2) edge node {} (u1);
\draw (w1) edge [bend left=10,densely dotted] node {} (v1);
\draw (w1) edge [bend right=10] node {} (v1);
\draw (w2) edge [bend left=10,densely dotted] node {} (v1);
\draw (w2) edge [bend right=10] node {} (v1);

\draw (v2) edge node {} (w1);

\draw (0.5,-1.9) node (k) {$k$};


\draw (2,-0.2) node (k) {\Large $\leadsto$};

\draw (4,-1.9) node (k) {$k$};

\draw (3.5,0) node[blacknode] [label=-90:$y$] (v1) {}
 ++(180:1cm) node[whitenode] [label=-90:$v_3$] (v3) {};

\draw (v1) edge node {} (v3);

\draw (4.5,0) node[blacknode] [label=-90:$u_1$] (u1) {}
 ++(0:0.5cm) node[whitenode] [label=-90:$u_2$] (u2) {};
\draw (u1) edge node {} (u2);

\draw (4,1) node[whitenode] [label=right:$w_1$] (w1) {}; 
\draw (4,-1) node[whitenode] [label=right:$w_2$] (w2) {}; 

\draw (w1) edge node {} (u1);
\draw (w2) edge node {} (u1);
\draw (w1) edge [bend left=10] node {} (v1);
\draw (w1) edge [bend right=10] node {} (v1);
\draw (w2) edge [bend left=10,densely dotted] node {} (v1);
\draw (w2) edge [bend right=10] node {} (v1);

\draw (w1) edge [bend left=10,densely dotted] node {} (v3);
\draw (w1) edge [bend right=10] node {} (v3);
\end{tikzpicture}
\quad
\begin{tikzpicture}[scale=\scalefactor]
\tikzstyle{whitenode}=[draw,circle,fill=white,minimum size=5pt,inner sep=0pt]
\tikzstyle{blacknode}=[draw,circle,fill=black,minimum size=5pt,inner sep=0pt]
\tikzstyle{texte}=[circle,minimum size=5pt,inner sep=0pt]
\tikzstyle{innerWhite} = [semithick, white,line width=0.5pt]
\draw (0,0) node[blacknode] [label=-90:$u_1$] (u1) {}
-- ++(0:0.5cm) node[blacknode] [label=-90:$u_2$] (u2) {}
-- ++(0:0.5cm) node[blacknode] [label=-90:$u_3$] (u3) {}
 ++(0:0.5cm) node[whitenode] [label=-90:$u_4$] (u4) {};

\draw (u3) edge node {} (u4);

\draw (0.75,1) node[whitenode] [label=right:$w_1$] (w1) {}; 
\draw (0.75,-1) node[whitenode] [label=right:$w_2$] (w2) {}; 

\draw (w1) edge [bend left=10,densely dotted] node {} (u1);
\draw (w1) edge [bend right=10] node {} (u1);
\draw (w2) edge [bend left=10,densely dotted] node {} (u1);
\draw (w2) edge [bend right=10] node {} (u1);
\draw (w1) edge [bend left=10,densely dotted] node {} (u2);
\draw (w1) edge [bend right=10] node {} (u2);
\draw (w2) edge node {} (u3);

\draw (0.75,-1.5) node (k) {\small{$\min\{\deg(u_1),\deg(u_2)\}=3$}};
\draw (0.75,-1.9) node (k) {$k$};

\draw (2.2,-0.2) node (k) {\Large $\leadsto$};

\draw (4,-1.9) node (k) {$k$};

\draw (4,0) node[blacknode] [label=left:$y$] (u2) {}
-- ++(0:0.5cm) node[blacknode] [label=-90:$u_3$] (u3) {}
 ++(0:0.5cm) node[whitenode] [label=-90:$u_4$] (u4) {};

\draw (u3) edge node {} (u4);

\draw (4,1) node[whitenode] [label=right:$w_1$] (w1) {}; 
\draw (4,-1) node[whitenode] [label=right:$w_2$] (w2) {}; 

\draw (w1) edge [bend left=10] node {} (u2);
\draw (w1) edge [bend right=10] node {} (u2);
\draw (w2) edge [bend left=10] node {} (u2);
\draw (w2) edge [bend right=10] node {} (u2);
\draw (w2) edge node {} (u3);
\draw (w1) edge node {} (u3);
\end{tikzpicture}
\label{fig:r:digon:2+2}
\caption{Reduction rules~\ref{r:digon:1+2:1} and~\ref{r:digon:3:1}.}
\label{fig:rules-digons-1}
\end{figure}
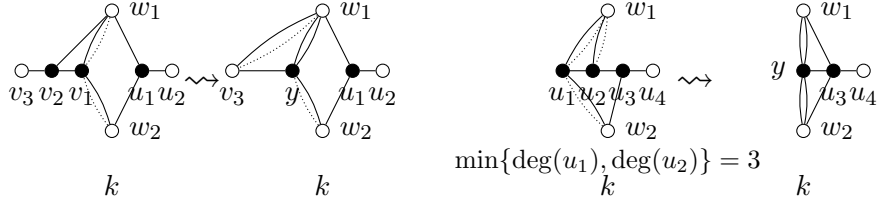

\rrule{r:digon:1+2:1}{Assume there are six vertices $v_1$, $v_2$, $v_3$, $u_1$, $u_2$, $w_1$, $w_2$, such that 
$N(u_1)= \{w_1,w_2,u_2\}$, $N(v_1)=\{w_1,w_2,v_2\}$, $N(v_2)= \{w_1,v_1,v_3\}$, and 
$\deg(v_2)=\deg(u_1)=3$.
Then contract the edge $v_1v_2$ to a new vertex $y$ and add an edge $w_1v_3$, as presented in Figure~\ref{fig:rules-digons-1} (left).}

\begin{lemma}
\label{lem:r:digon:1+2:1}
Rule~\ref{r:digon:1+2:1} is correct.
\end{lemma}

\begin{proof}
Let $G'$ be the graph obtained from a graph $G$ by a single application of Rule~\ref{r:digon:1+2:1}.  
Note that $G'$ is a result of a gadget replacement $(X,Y,E_I)$ with $X=\{v_1,v_2,u_1\}$ and $Y=\{y,u_1\}$

Let $S$ be a feedback vertex set of size $k$ in $G$. 
We claim that there is a feedback vertex set $S'$ in $G'$ of size at most $k$.
If $|S\cap\{w_1,w_2,v_1,v_2,u_1\}|\ge 2$, then by Lemma~\ref{lem:fvs-transfer} we see that $S'=(S\setminus\{w_1,w_2,v_1,v_2,u_1\})\cup\{w_1,w_2\}$ works. 
Hence we can assume $|S\cap\{w_1,w_2,v_1,v_2,u_1\}|\le 1$.
Then $S'\cap\{w_1,w_2,v_1,v_2,u_1\} = \{w_1\}$ or $S'\cap\{w_1,w_2,v_1,v_2,u_1\} = \{v_1\}$ to hit the triangle $v_1v_2w_1$ and the quadrangle $v_1w_1u_1w_2$.
In the prior case we pick $S=S'$.
Then $R_{G[N_G[X]]-S,\{w_1,w_2,u_2,v_3\}}$ and $R_{G'[N_{G'}[Y]]-S',\{w_1,w_2,u_2,v_3\}}$ are the same relations since the $v_3v_2v_1$ path in $G$ corresponds to the $v_3y$ edge in $G'$, so Lemma~\ref{lem:fvs-transfer} applies.
We are left with the case $S'\cap\{w_1,w_2,y_1,u_1\} = \{v_1\}$.
Then we pick $S=(S'\setminus\{v_1\})\cup\{y\}$.
Again, $R_{G[N_G[X]]-S,\{w_1,w_2,u_2,v_3\}}=R_{G'[N_{G'}[Y]]-S',\{w_1,w_2,u_2,v_3\}}$ since the $v_3v_2w_1$ path in $G$ corresponds to the $v_3w_1$ edge in $G'$, so Lemma~\ref{lem:fvs-transfer} applies.

Let $S'$ be a feedback vertex set of size $k$ in $G'$. 
We claim that there is a feedback vertex set $S$ in $G$ of size at most $k$.
If $|S'\cap\{w_1,w_2,y_1,u_1\}|\ge 2$, then by Lemma~\ref{lem:fvs-transfer} we see that $S=(S'\setminus\{y_1,y_2\})\cup\{w_1,w_2\}$ works. 
Hence we can assume $|S'\cap\{w_1,w_2,y_1,u_1\}|\le 1$.
Then $S'\cap\{w_1,w_2,y_1,u_1\} = \{w_1\}$ or $S'\cap\{w_1,w_2,y_1,u_1\} = \{y\}$ to hit the digon $w_1y$.
In the prior case we pick $S=S'$.
Then $R_{G[N_G[X]]-S,\{w_1,w_2,u_2,v_3\}}=R_{G'[N_{G'}[Y]]-S',\{w_1,w_2,u_2,v_3\}}$ since the $v_3v_2v_1$ path in $G$ corresponds to the $v_3y$ edge in $G'$, so Lemma~\ref{lem:fvs-transfer} applies.
We are left with the case $S'\cap\{w_1,w_2,y_1,u_1\} = \{y\}$.
Then we pick $S=S'\setminus\{y\}\cup\{v_1\}$.
Again, $R_{G[N_G[X]]-S,\{w_1,w_2,u_2,v_3\}}=R_{G'[N_{G'}[Y]]-S',\{w_1,w_2,u_2,v_3\}}$ since the $v_3v_2w_1$ path in $G$ corresponds to the $v_3w_1$ edge in $G'$, so Lemma~\ref{lem:fvs-transfer} applies.
\end{proof}

\rrule{r:digon:3:1}{Assume $u_1u_2u_3u_4$ is an induced path such that for two vertices $w_1$, $w_2$ outside the path, $N(u_1)=\{u_2,w_1,w_2\}$, $N(u_2)=\{u_1,u_3,w_1\}$ and $N(u_3)=\{u_2,u_4,w_2\}$, $\deg(u_3)=3$, and $\min\{\deg(u_1), \deg(u_2)\}=3$.
Then replace $G[\{u_1,u_2,u_3,w_1,w_2\}]$ with the gadget presented in Figure~\ref{fig:rules-digons-1} (right), i.e., remove $u_1$ and $u_2$ and add a vertex $y$, edges $yu_3$ and $u_3w_1$, and double edges $yw_1$ and $yw_2$.}

\begin{lemma}
\label{lem:r:digon:3:1}
Rule~\ref{r:digon:3:1} is correct.
\end{lemma}
\begin{proof}
Let $G'$ be the graph obtained from a graph $G$ by a single application of Rule~\ref{r:digon:3:1}.  
Note that $G'$ is a result of a gadget replacement $(X,Y,E_I)$ with $X=\{u_1,u_2\}$ and $Y=\{y\}$

Let $S$ be a solution of $(G,k)$. 
If $|S\cap\{w_1,w_2,u_1,u_2,u_3\}|\ge 2$ or $S\cap\{w_1,w_2,u_1,u_2,u_3\}=\{w_1\}$ we proceed as in the proof of Lemma~\ref{lem:r:digon:1+2:1}.
Otherwise, to hit the triangle $w_1u_1u_2$, $S\cap\{w_1,w_2,u_1,u_2,u_3\}$ equals either $\{u_1\}$ or $\{u_2\}$.
In both cases, $R_{G[N_G[X]]-S,\{w_1,w_2,u_3\}}$ has exactly one equivalence class $\{w_1,w_2,u_3\}$.
We observe that for $S'=S \setminus \{u_1,u_2\} \cup \{y\}$ the relation $R_{G'[N_{G'}[Y]]-S',\{w_1,w_2,u_3\}}$ has also one equivalence class, so by  Lemma~\ref{lem:fvs-transfer}, $S'$ is a solution of $(G',k)$.

Let $S'$ be a feedback vertex set of size $k$ in $G'$. 
If $|S'\cap\{w_1,w_2,y,u_3\}|\ge 2$ we proceed as in the proof of Lemma~\ref{lem:r:digon:1+2:1}.
Otherwise, $S'\cap\{w_1,w_2,u_1,u_2,u_3\}=\{y\}$.
If $\deg_G(u_1)=3$ then we put $S=S'\setminus\{y\}\cup\{u_2\}$, and otherwise $S=S'\setminus\{y\}\cup\{u_1\}$.
Note that $G[\{w_1,w_2,u_1,u_2,u_3\}]-S$ is a forest, since $\min\{\deg(u_1), \deg(u_2)\}=3$.
Moreover, $R_{G[N_G[X]]-S,\{w_1,w_2,u_3\}}$ and $R_{G'[N_{G'}[Y]]-S',\{w_1,w_2,u_3\}}$ are the same (total) relation, so Lemma~\ref{lem:fvs-transfer} applies and $(S,k)$ is a solution of $(G,k)$.
\end{proof}

\rrule{r:meta}{
Let $A \subseteq V(G)$ and let $w_1$ and $w_2$ be two vertices in $G$, $w_1,w_2\not\in A$.
If
$(i)$ no cycle in $G\setminus\{w_1,w_2\}$ intersects $A$, and
$(ii)$ there is a subgraph $Q \subseteq G[A\cup\{w_1, w_2\}]$ such for every vertex $x\in V(Q)\setminus\{w_1\}$, we have $\deg_Q(x) \le |E(Q)|-|A|-1$,
then remove $w_1$ and decrease $k$ by 1.
}

\begin{lemma}\label{lem:r:meta}
Rule~\ref{r:meta} is correct.
\end{lemma}
\begin{proof}
Let $G'$ be the graph obtained from a graph $G$ by a single application of Rule~\ref{r:meta}, i.e., $G'=G-w_1$. 
Let $S$ be a feedback vertex set of size $k-1$ in $G'$. 
Then every cycle in $G-S$ contains $w_1$, so $S\cup\{w_1\}$ is a feedback vertex set of size $k$ in $G$.

Let $S$ be a feedback vertex set of size $k$ in $G$. 
If $w_1 \in S$, then clearly $S \setminus \{w_1\}$ is a solution of the instance $(G',k-1)$. 
Hence assume $w_1 \not\in S$. 
We claim that $|S\cap V(Q)|\ge 2$.
Assume the contrary, i.e., $|S\cap V(Q)|\le 1$.
Since $Q-S$ is a forest, 
\begin{equation}
 \label{eq:forest}
 |E(Q-S)| \le |V(Q-S)| - 1 = |V(Q)| - |S\cap V(Q)| - 1 = |A| + 1 - |S\cap V(Q)|.
\end{equation}
On the other hand, by the degree bound, and because $w_1\not\in S$ and $|S\cap V(Q)|\le 1$, 
\begin{equation}
 \label{eq:degree}
 |E(Q-S)| \ge |E(Q)| - (|E(Q)|-|A|-1)|S\cap V(Q)|.
\end{equation}
By~\eqref{eq:forest} and~\eqref{eq:degree}, $|A|+1 \ge |E(Q)| - (|E(Q)|-|A|-2)|S\cap V(Q)|$.
Since $|S\cap V(Q)|\le 1$ this implies $|A|+1 \ge |E(Q)| - (|E(Q)|-|A|-2) = |A|+2$, a contradiction.
It follows that $|S\cap V(Q)|\ge 2$.
Then $S'=S\setminus \{u,v_1,v_2,v\} \cup \{w_1,w_2\}$ is of size at most $k$.
Moreover, $S'$ is a feedback vertex set in $G$, since $S$ is a feedback vertex set and by $(i)$.
Again, this implies that $S' \setminus \{w_1\}$ is a solution of the instance $(G',k-1)$, as required. 

\end{proof}

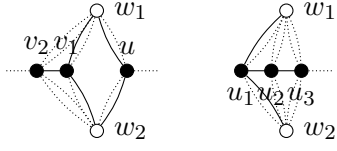
\begin{figure}[t]
\centering
\begin{tikzpicture}[scale=\scalefactor]
\tikzstyle{whitenode}=[draw,circle,fill=white,minimum size=5pt,inner sep=0pt]
\tikzstyle{blacknode}=[draw,circle,fill=black,minimum size=5pt,inner sep=0pt]
\tikzstyle{texte}=[circle,minimum size=5pt,inner sep=0pt]
\tikzstyle{innerWhite} = [semithick, white,line width=0.5pt]
\draw (0,0) node[blacknode] [label=90:$v_1$] (v1) {}
-- ++(180:0.5cm) node[blacknode] [label=90:$v_2$] (v2) {}
 ++(180:0.7cm) node (v3) {};

\draw[densely dotted] (v2) edge node {} (v3);

\draw (1,0) node[blacknode] [label=90:$u$] (u1) {}
 ++(0:0.7cm) node (u2) {};
\draw[densely dotted] (u1) edge node {} (u2);

\draw (0.5,1) node[whitenode] [label=right:$w_1$] (w1) {}; 
\draw (0.5,-1) node[whitenode] [label=right:$w_2$] (w2) {}; 

\draw (w1) edge [bend left=10,densely dotted] node {} (u1);
\draw (w1) edge [bend right=10] node {} (u1);
\draw (w2) edge [bend left=10,densely dotted] node {} (u1);
\draw (w2) edge [bend right=10] node {} (u1);
\draw (w1) edge [bend left=10,densely dotted] node {} (v1);
\draw (w1) edge [bend right=10] node {} (v1);
\draw (w2) edge [bend left=10,densely dotted] node {} (v1);
\draw (w2) edge [bend right=10] node {} (v1);

\draw (v2) edge [bend right=10,densely dotted] node {} (w1);
\draw (v2) edge [bend left=10,densely dotted] node {} (w1);
\draw (v2) edge [bend left=10,densely dotted] node {} (w2);
\draw (v2) edge [bend right=10, densely dotted] node {} (w2);

\end{tikzpicture}
\quad
\begin{tikzpicture}[scale=\scalefactor]
\tikzstyle{whitenode}=[draw,circle,fill=white,minimum size=5pt,inner sep=0pt]
\tikzstyle{blacknode}=[draw,circle,fill=black,minimum size=5pt,inner sep=0pt]
\tikzstyle{texte}=[circle,minimum size=5pt,inner sep=0pt]
\tikzstyle{innerWhite} = [semithick, white,line width=0.5pt]
\draw (0,0) node[blacknode] [label=-90:$u_1$] (u1) {}
-- ++(0:0.5cm) node[blacknode] [label=-90:$u_2$] (u2) {}
-- ++(0:0.5cm) node[blacknode] [label=-90:$u_3$] (u3) {}
 ++(0:0.7cm) node (u4) {};

\draw[densely dotted] (u3) edge node {} (u4);

\draw (0.75,1) node[whitenode] [label=right:$w_1$] (w1) {}; 
\draw (0.75,-1) node[whitenode] [label=right:$w_2$] (w2) {}; 

\draw (w1) edge [bend left=10,densely dotted] node {} (u1);
\draw (w1) edge [bend right=10] node {} (u1);
\draw (w2) edge [bend left=10,densely dotted] node {} (u1);
\draw (w2) edge [bend right=10] node {} (u1);
\draw (w1) edge [bend left=10,densely dotted] node {} (u2);
\draw (w1) edge [bend right=10,densely dotted] node {} (u2);
\draw (w2) edge [bend left=10,densely dotted] node {} (u2);
\draw (w2) edge [bend right=10,densely dotted] node {} (u2);
\draw (w1) edge [bend left=10,densely dotted] node {} (u3);
\draw (w1) edge [bend right=10,densely dotted] node {} (u3);
\draw (w2) edge [bend left=10,densely dotted] node {} (u3);
\draw (w2) edge [bend right=10,densely dotted] node {} (u3);

\end{tikzpicture}
\caption{Configurations in lemmas~\ref{lem:digon:1+2:2} and~\ref{lem:digon:3:2}.}
\label{fig:rules-2}
\end{figure}

Rule~\ref{r:meta} is not used directly in our algorithm, because it seems impossible to detect it in $O(n)$ time.
However, to get the claimed kernel size we need just two special cases of Rule~\ref{r:meta}, which are stated in lemmas ~\ref{lem:digon:1+2:2} and \ref{lem:digon:3:2} below.

\begin{lemma}
\label{lem:digon:1+2:2} 
Assume there are five vertices $v_1$, $v_2$, $u$, $w_1$, $w_2$ such that 
$N(v_1)=\{v_2,w_1,w_2\}$, $\{w_1,w_2\}\subseteq N(u)$,  
there is at most one edge incident to $v_2$ and a vertex outside $\{w_1,w_2,v_1\}$, and
there is at most one edge incident to $u$ and a vertex outside $\{w_1,w_2\}$.
Then Rule~\ref{r:meta} applies.
\end{lemma}

\begin{proof}
It is easy to see that condition $(i)$ of Rule~\ref{r:meta} is satisfied. We proceed to condition $(ii)$.
Since Rule~\ref{r:deg2} does not apply, $v_2$ is adjacent to $w_1$ or $w_2$; by symmetry assume the former.
Let $A=\{u,v_1,v_2\}$.
We build $E(Q)$ as follows.
We start with $E(Q)=\{v_2w_2,v_2v_1,v_1w_1,v_1w_2,uw_1,uw_2\}$.
Since Rule~\ref{r:digon:1+2:1} does not apply, $v_2w_2 \in E$ or one of $\{v_2w_1,uw_1,uw_2\}$ is a double edge.
Hence, we add $v_2w_2$ or another copy of one of $\{v_2w_1,uw_1,uw_2\}$ to $E(Q)$, respectively.
Note that for every $x\in V(Q)\setminus\{w_1\}$ we have $\deg_Q(x)\le 3 = |E(Q)|-|A|-1$, as required.
\end{proof}

\begin{lemma}
\label{lem:digon:3:2} 
Assume there are five vertices $u_1$, $u_2$, $u_3$, $w_1$, $w_2$ such that  $N(u_1)=\{w_1,w_2,u_2\}$, $\{u_1,u_3\}\subseteq N(u_2)\subseteq \{w_1,w_2,u_1,u_3\}$, and there is at most one edge incident to $u_3$ and a vertex outside $\{w_1,w_2,u_2\}$.
Moreover, the edges $v_1v_2$ and $v_2v_3$ are simple.
Then Rule~\ref{r:meta} applies.
\end{lemma}

\begin{proof}
It is easy to see that condition $(i)$ of Rule~\ref{r:meta} is satisfied. We proceed to condition $(ii)$.
Let $A=\{u_1,u_2,u_3\}$.
We build $E(Q)$ as follows.
We start with $E(Q)=\{u_1w_1,u_1w_2,u_1u_2,u_2u_3\}$.
There are some cases to consider.
Since Rule~\ref{r:deg2} does not apply, $u_2w_1\in E$ or $u_2w_2\in E$.

\mycase{1:} $u_2w_1, u_2w_2\in E$. 
Then, since Rule~\ref{r:gamma} does not apply to the $w_1u_1w_2u_2$ cycle, $u_1w_1$ or $u_1w_2$ is a double edge.
Moreover, since Rule~\ref{r:deg2} does not apply, $u_3w_1\in E$ or $u_3w_2\in E$.
We add to $E(Q)$ edges $u_2w_1$, $u_2w_2$, either $u_3w_1$ or $u_3w_2$ (but not both), and the second copy of either $u_1w_1$ or $u_1w_2$ (but not both).
Then $|E(Q)|=8$ and $\max_{x\in V(Q)\setminus\{w_1\}}\deg_Q(x) = 4 = |E(Q)|-|A|-1$, so $(ii)$ holds.

\mycase{2:} Exactly one of $u_2w_1$ and $u_2w_2$ is an edge; by symmetry assume $u_2w_1\in E$ and $u_2w_2\not\in E$. 
Since Rule~\ref{r:gamma} does not apply, $u_3w_1\not\in E$.
And then since Rule~\ref{r:deg2} does not apply, $u_3w_2\in E$.

\mycase{2.1:} $u_3w_2$ is a double edge. 
We add to $E(Q)$ edge $u_2w_1$, and both copies of $u_3w_2$.
Then $|E(Q)|=7$ and $\max_{x\in V(Q)}\deg_Q(x) = 3 = |E(Q)|-|A|-1$, so $(ii)$ holds.

\mycase{2.2:} $u_3w_2$ is a simple edge. 
Since Rule~\ref{r:digon:3:1} does not apply, $\deg_G(u_1)\ge 4$ and $\deg_G(u_2)\ge 4$.
We add to $E(Q)$ edges $u_2w_1$ and $u_3w_2$, exactly one edge incident to $u_1$ which is not yet in $E(Q)$ and exactly one edge incident to $u_2$ which is not yet in $E(Q)$.
Then $|E(Q)|=8$ and $\max_{x\in V(Q)\setminus\{w_1\}}\deg_Q(x) = 4 = |E(Q)|-|A|-1$, so $(ii)$ holds.
\end{proof}


The following rule was shown to be correct by Abu-Khzam and Khuzam in~\cite{ak12}. 

\rrule{r:indpath6}{Assume there is an induced path with endpoints $u$ and $v$ and with six internal vertices $v_1,\ldots,v_{6}$ such that for some vertices $w_1$, $w_2$ outside the path $N(\{v_1,\ldots,v_{6}\})\setminus\{u,v\} = \{w_1,w_2\}$.
If $|N(w_1)\cap \{v_1,\ldots,v_{6}\}|\ge|N(w_2)\cap \{v_1,\ldots,v_{6}\}|$, then remove $w_1$ and decrease $k$ by one.}

\medskip 

In~\cite{ak12} it was assumed that when Rule~\ref{r:indpath6} described above is applied, $G$ does not contain an induced path $v_1,\ldots,v_5$ such that for some vertex $w$, we have $N(v_2,v_3,v_4)\setminus\{v_1,v_5\} = \{w\}$.  In our algorithm this is guaranteed by Rule~\ref{r:gamma} (slightly more general than their Rule 6). We are able to extend Rule~\ref{r:indpath6} as follows.

\newcounter{indpath5}
\setcounter{indpath5}{\value{theorem}}
\begin{lemma}
 \label{lem:indpath5}
 Assume there is an induced path with endpoints $u$ and $v$ and with five internal vertices $v_1,\ldots,v_5$ such that for some vertices $w_1$, $w_2$ outside the path $N(\{v_1,\ldots,v_5\})\setminus\{u,v\} = \{w_1,w_2\}$.
 Then there is an instance $(G',k')$ with $|V(G')|<|V(G)|$ such that $(G,k)$ is a yes-instance iff $(G',k')$ is a yes-instance and $k'\le k$.
\end{lemma}

The proof of Lemma~\ref{lem:indpath5} involves five more rules and is quite technical; we defer it to Section~\ref{sec:5path}.
We stress here that Lemma~\ref{lem:indpath5} is not crucial for getting a substantial improvement of the kernel size.
Indeed, if one uses Rule~\ref{r:indpath6} instead of Lemma~\ref{lem:indpath5}, the resulting kernel is of size at most $15k-28$ (see Section~\ref{sec:analysis}). 
Let us also remark that by the analysis in Section~\ref{sec:analysis}, if someone manages to exclude paths described in Lemma~\ref{lem:indpath5} with only {\em four} internal vertices, the kernel size decreases further to $11k-20$.

To complete the algorithm we need a final rejecting rule which is applied when the resulting graph is too big. In Section~\ref{sec:analysis} we prove that Rule~\ref{r:reject} is correct.

\rrule{r:reject}{If the graph has more than $13k-24$ vertices, return a trivial no-instance (conclude that there is no feedback vertex set of size $k$ in $G$).} 

\section{The size bound}
\label{sec:analysis}

In this section we prove the following theorem.

\begin{theorem}
\label{thm:size-bound}
Let $G$ be a planar graph such that rules 1--\ref{r:digon:3:1} do not apply and $G$ does not contain the configurations described in lemmas~\ref{lem:digon:1+2:2} and \ref{lem:digon:3:2}. 
Assume also that for every induced path $P$ with endpoints $u$ and $v$ and with $\ell$ internal vertices $v_1$, \ldots, $v_{\ell}$ the internal vertices have at least three neighbors outside the path, i.e., $|N(\{v_1,\ldots,v_{\ell}\})\setminus\{u,v\}|\ge 3$.
If there is a feedback vertex set of size $k$ in $G$, then $|V(G)|\le (2\ell+4)k- (4\ell+6)$.
\end{theorem}

Let $S$ be a feedback vertex set of size $k$ in $G$ (i.e., a ``solution''), and let $F$ be the forest induced by $V(G)\setminus S$.
Denote the set of vertices of $F$ by $V_F=V(G)\setminus S$.
We call the vertices in $S$ {\em solution vertices} and the vertices in $V_F$ {\em forest vertices}. 


\heading{A partition of $V_F$}
Now we define some subsets of $V_F$.
Let $I_{2}, I_{3^+} \subseteq V_F$ denote the vertices whose degree in $F$ is two or at least three, respectively.
The leaves of $F$ are further partitioned into two subsets.
Let $L_2$ and $L_{3^+}$ be the leaves of $F$ that have two or at least three solution neighbors, respectively. 
By rules~\ref{r:deg1} and~\ref{r:deg2} all the vertices in $G$ have degree at least $3$.
Hence, if a leaf of $F$ has fewer than two solution neighbors, Rule~\ref{r:deg3-double} or Rule~\ref{r:triple} applies. 
It follows that every leaf of $F$ belongs to $L_2 \cup L_{3^+}$. 
This proves claim $(i)$ of Lemma~\ref{lem:prop-G} below. 
\begin{lemma}
\label{lem:prop-G}
Graph $G$ satisfies the following properties.
 \begin{enumerate}[$(i)$]
  \item The sets $I_2$, $I_{3^+}$, $L_2$, $L_{3^+}$ form a partition of $V_F$. 
  \item For every pair $u$, $v$ of solution vertices there are at most two vertices $x,y\in L_2$ such that $N(x)\cap S = N(y)\cap S = \{u,v\}$.
  \item Every vertex of $G$ is of degree at least three.
  \item Every face of $G$ is of length at least two.
 \end{enumerate}
\end{lemma} 

Claim $(ii)$ follows from the fact that Rule~\ref{r:3deg3} does not apply to $G$.
Claim $(iii)$ follows because rules~\ref{r:deg1} and~\ref{r:deg2} do not apply to $G$ and Claim $(iv)$ by Rule~\ref{r:loop}.

\heading{The inner forest}
Let $F_I$ be the forest on the vertex set $I_{3^+}\cup L_{3^+}$ such that $uv\in E(F_I)$ iff for some integer $i\ge 0$, 
there is a path $ux_1\cdots x_iv$ in forest $F$ such that $u,v\in I_{3^+}\cup L_{3^+}$ and for every $j=1,\ldots,i$, vertex
$x_i$ belongs to $I_2$.

\heading{Three sets of short chains}
A path in $F$ consisting of vertices from $I_2 \cup L_2$ will be called a {\em chain}.
A chain is maximal if it is not contained in a bigger chain.
In what follows we introduce three sets of (not necessarily maximal) chains, denoted by $CL_2$, $C_{2^-}$ and $C_{3^+}$.
We will do it so that each vertex in $I_2$ belongs to {\em at least one} chain from these sets of chains.

For every vertex $x \in L_2$, we consider the maximal chain $(y_1,\ldots,y_p)$ of degree $2$ vertices in $F$ such that $y_1$ is adjacent to $x$ and no $y_i$ has a solution neighbor outside $N_G(x)\cap S$. 
Then the chain $(x,y_1,\ldots,y_p)$ is an element of $CL_2$. 
Note that $L_2 \subseteq V(CL_2)$.

Chains of $C_{2^-}$ and $C_{3^+}$ are defined using the following algorithm.
We consider maximal chains in $F$, one by one (note that all maximal chains are vertex-disjoint).
Let $c=(x_1,x_2,\ldots,x_p)$ be a maximal chain. 
The vertices of $c$ are ordered so that if $\{x_1,x_p\}\cap L_2 \ne \emptyset$, then $x_p\in L_2$.
Using vertices of $c$ we form disjoint bounded length chains and put them in the sets $C_{2^-}$ and $C_{3^+}$ as follows.
Assume that for some $i<p$ the vertices of a prefix $(x_1,x_2,\ldots,x_i)$ have been already partitioned into such chains (in particular $i=0$ if we begin to process $c$). There are three cases to consider.

Consider a shortest chain $c_i=(x_{i+1},\ldots,x_j)$ such that the vertices of $c_i$ have at least three solution neighbors, i.e., $|S\cap N(\{x_{i+1},\ldots,x_j\})|\ge 3$.
If the chain $c_i$ exists, we put it in $C_{3^+}$, and we proceed to the next vertices of $c$.
Otherwise we consider the chain $c'_i=(x_{i+1},\ldots,x_p)$.
Note that vertices of $c'_i$ have at most two solution neighbors.

If $x_p\in I_2$, then we add the chain $c'_i$ to $C_{2^-}$ and we finish processing $c$.
Note that then $x_p$ is adjacent to a vertex $u\in L_{3^+}\cup I_{3^+}$ (otherwise $c$ is not maximal, as we can extend it by a vertex in $L_2$). 
Moreover, because of the order of the vertices in $c$, we know that $x_1\not\in L_2$.
It follows that $x_1$ is also adjacent to a vertex $v\in L_{3^+}\cup I_{3^+}$.
Hence, $uv\in E(F_I)$. We assign chain $c'_i$ to edge $uv$.

If $x_p\in L_2$, then we do not form a new chain and we finish processing $c$.
Note, however, that the vertices $\{x_{i+1},\ldots,x_p\}\cap I_2$ belong to a chain in $CL_2$.

Note also that some vertices of the first chain $c_0$ can belong to two chains, one in $C_{3^+}$ and one in $CL_2$.

Let us summarize the main properties of the construction.

\begin{lemma}
\label{lem:construction}
The following properties hold:
 \begin{enumerate}[$(i)$]
  \item Every vertex from $I_2$ belongs to a chain in $CL_2$, $C_{2^-}$ or $C_{3^+}$.
  \item Every chain in $CL_2 \cup C_{2^-}$ has at most two solution neighbors.
  \item Every chain in $C_{3^+}$ has at least three solution neighbors.
  \item Every chain in $C_{2^-}$ is assigned to a different edge of inner forest $F_I$.
  \item Every chain in $C_{2^-}\cup CL_2$ has at most $\ell-1$ vertices.
  \item Every chain in $C_{3^+}$ has at most $\ell$ vertices.
 \end{enumerate}

\end{lemma}


\heading{A solution graph $H_S$}
Let us introduce a new plane multigraph $H_S=(S,E_S)$.
Since the vertices of $H_S$ are the solution vertices we call it a {\em solution graph}. 
From now on, we fix a plane embedding of $G$.
The vertices of $H_S$ are embedded in the plane exactly in the same points as in $G$.
The edge multiset $E_S$ is defined as follows.
For every triple $(u,x,v)$ such that $u,v\in S$, $x\in L_2$ and there is a path $uxv$ in $G$, we put an edge $uv$ in $E_S$.
Moreover, the edge $uv$ is embedded in the plane exactly as one of the corresponding paths $uxv$ (note that there can be up to four such paths if some edges are double).
Note that by Lemma~\ref{lem:prop-G}$(ii)$, every edge of $H_S$ has multiplicity at most two.

The set of faces of $H_S$ is denoted by $F_S$. By $F_{S,2}$ we denote its subset with the faces of length two, while $F_{S,3+}$ are the remaining faces.
Note that there are no faces of length 1 in $H_S$.

\begin{lemma}\label{lem:L23E}
We have $|V(CL_2)| \leq 2 (|E_S|-|F_{S,2}|)$.
\end{lemma}
\begin{proof}
By the definition, for every vertex $x\in L_2$ there is a corresponding edge $uv\in E_S$, where $N_G(x)\cap S = \{u,v\}$.
Also, for every chain $c$ in $CL_2$ there is a corresponding vertex $x\in L_2$, and thus a corresponding edge $uv\in E_S$.
We {\em assign} $x$, $c$ and the vertices of $c$ to the pair $\{u,v\}$.

Consider an arbitrary pair $u$, $v$ such that $uv\in E_S$. Note that there are exactly $|E_S|-|F_{S,2}|$ such pairs.
We claim that there are at most two elements in $V(CL_2)$ assigned to the pair $\{u,v\}$.
Indeed, by Lemma~\ref{lem:prop-G}$(ii)$, there are at most two vertices in $L_2$ assigned to $\{u,v\}$.
If there are no such vertices, no chain in $CL_2$ is assigned to $\{u,v\}$, so the claim holds.
If there is exactly one vertex $x\in L_2$ assigned, there is exactly one chain $c\in CL_2$ assigned. 
By Lemma~\ref{lem:r:digon:3:1}, chain $c$ has at most two vertices, so the claim holds.
Finally, if there are exactly two vertices $x,y\in L_2$ assigned, there are exactly two chains $c_x$ and $c_y $ assigned.
By Lemma~\ref{lem:r:digon:1+2:1} we have $|V(c_x)|=|V(c_y)|=1$. This concludes the proof. 
\end{proof}

\heading{Maximality}
In what follows we assume that graph $G$ is {\em maximal}, meaning that one can add neither an edge to $E(G)$ nor a vertex to $L_2$ obtaining a graph $G'$ such that $S$ is still a feedback vertex set of $G'$ and all the claims of lemmas~\ref{lem:prop-G}, \ref{lem:construction} and \ref{lem:L23E} hold.
Note that the number of $L_2$-vertices which can be added to $G$ is bounded, since each such vertex corresponds to an edge in $H_S$, 
and $H_S$ has at most $6|S|$ edges as a plane multigraph with edge multiplicity at most two.
Similarly, once the set of $L_2$-vertices is maximal, and hence the vertex set of $G$ is fixed, the number of edges which can be added to $G$ is bounded by $6|V(G)|$. 
It follows that such a maximal supergraph of $G$ exists.
Clearly, it is sufficient to prove Theorem~\ref{thm:size-bound} only in the case when $G$ is maximal.

\begin{lemma}\label{cl:solutionconnected}
The planar graph $H_S$ is connected.
\end{lemma}

\begin{proof}
%
Assume now for contradiction that there is a partition $S=S_1 \cup S_2$ such that there is no edge in $H_S$ between a vertex of $S_1$ and a vertex of $S_2$. 

Every face of $G$ is incident to at least one vertex of $S$, for otherwise the boundary of the face does not contain a cycle, a contradiction.
Assume that a face $f$ of $G$ contains a solution vertex $u_1$ in $S_1$ and a solution vertex $u_2$ in $S_2$.
Then we can add a vertex $x$, two edges $xu_1$ and two edges $xu_2$.
Note that $S$ is still a feedback vertex set in the new graph; in particular now $x\in L_2$.
In the new graph there are no more vertices in $L_2$ adjacent to both $u_1$ and $u_2$ because of our assumption that $S_1$ and $S_2$ are not connected by an edge in $H_S$, so Lemma~\ref{lem:prop-G}$(ii)$ holds.
Moreover,  $|V(CL_2)|$ was increased by one and $|E_S|-|F_{S,2}|$ was also increased by one, so Lemma~\ref{lem:L23E} holds.
The other claims of lemmas~\ref{lem:prop-G} and \ref{lem:construction} trivially hold, so $F$ is not maximal, a contradiction.

Let $\Ff_1$ and $\Ff_2$  be the collections of faces of $G$ containing a vertex in $S_1$, or in $S_2$, respectively. 
We have shown above that $\Ff_1\cup \Ff_2$ is a partition of the set of all the faces of $G$. 
Let $V_1$ and $V_2$ denote the sets of vertices incident to a face in $\Ff_1$, or in $\Ff_2$, respectively. 
Note that $V_1\cap V_2 \ne \emptyset$, since there must be two neighboring faces, one in $\Ff_1$ and the other in $\Ff_2$.
Let $x\in V_1\cap V_2$. Since faces of $G$ are of length at least two, $x$ has in $G$ at least two neighbors in $V_1\cap V_2$.
It follows that $G[V_1\cap V_2]$ has minimum degree two, so $G[V_1\cap V_2]$ contains a cycle. However, $(V_1\cap V_2) \cap S = \emptyset$, since $\Ff_1$ and $\Ff_2$ are disjoint. Hence $V_1\cap V_2 \subseteq F$, a contradiction.

\end{proof}

\heading{Bounding the number of forest vertices in a face of $H_S$}
For a face $f$ of $H_S$ and a set of vertices $A\subseteq V(G)$ we define $A^f$ as the subset of $A$ of vertices which are embedded in $f$ or belong to the boundary of $f$. 
Note that all vertices of every chain belong to the same face $f$ of $H_S$.
When $C$ is a set of chains, by $C^f$ we denote the subset of chains of $C$ which lie in $f$, i.e., $C^f=\{c \in C\ :\ V(c) \subseteq V(G)^f\}$.

\begin{lemma}\label{cl:df2}
For every face $f$ of $H_S$, it holds that $|L^f_{3^+}|+|I^f_{3^+}|+|C^f_{3^+}|\leq d(f)-2$.
\end{lemma}
\begin{proof}
First we note that the forest $F^f$ is in fact a tree. Indeed, if $F^f$ has more than one component, we can add an edge between two solution vertices on the boundary of $f$ preserving planarity, what contradicts the assumed maximality.

Consider a plane subgraph $A$ of $G$ induced by $V(G)^f$, i.e., we take the plane embedding of $G$ and we remove the vertices outside $V(G)^f$.
Then we can define graph $A_S$, analogously to $H_S$. We treat $f$ as a face of $A_S$.
Let $u_1u_2\cdots u_{d(f)}u_1$ be the facial walk of $f$.

\ignore{
Let $u_1,u_2,\ldots,u_{d(f)},u_1$ be the facial walk of $f$.
(Note that for all $i=1,\ldots,d(f)$ we have $u_i\in S$.)
Note that this walk does not need to be a simple cycle, since $H_S$, and consequently the boundary of $f$, does not need to be 2-connected.
However, we can get an equivalent graph where the facial walk of $f$ is a cycle, as follows.
If the boundary of $f$ contains a cutvertex, there is also a cutvertex $u_i$ incident to the outer face of $G_1$.
Then for some $j\ne i$, we have that $u_j=u_i$ and for all indices $j'=i+1,\ldots,j-1$ taken modulo $d(f)$ it holds that $u_{j'}\ne u_i$.
Order all edges of $A$ incident to $u_i$ in the clockwise order around $u_i$ in the considered embedding; the other endpoints form the sequence $u_{i-1},w_1,\ldots,w_p,u_{i+1},u_{j-1},z_1,\ldots,z_q,u_{j+1}$.
Replace in $A$ the vertex $u_i$ by two vertices $x$ and $y$ and connect $x$ to $u_{i-1},w_1,\ldots,w_p,u_{i+1}$ and $y$ to $u_{j-1},z_1,\ldots,z_q,u_{j+1}$.
The resulting graph is plane, since we can obtain its embedding by a minor modification of the previous graph (see Fig.~\ref{fig:topology}).
We repeat this transformation as long as there are cutvertices in graph $A$. 
\begin{figure}[t]
\centering
\includegraphics[width=5cm]{topology.pdf}
\caption{\label{fig:topology}}
\end{figure}
}

Consider an arbitrary vertex $x$ of $I^f_{3^+}$. 
Let $T_1,\ldots,T_r$ be the $r$ trees obtained from the tree $T$ in $F$ containing $x$ after removing $r$ from $T$.
Then $r\ge 3$ since $x$ has at least three neighbors in $T$.
By planarity, there are $2r$ indices $b_1,e_1,b_2,e_2,\ldots,b_r,e_r$ such that for every $i=1,\ldots,r$ 
\[\{u_{b_i},u_{e_i}\} \subseteq N(V(T_t)) \cap \{u_1,\ldots,u_{d(f)}\} \subseteq \{u_{b_i},u_{b_i+1},\ldots,u_{e_i}\}.\]
Then, for every $j\in\{b_1,b_2,\ldots,b_r\}$ there is an edge $xu_j$, for otherwise we can add it in the current plane embedding, contradicting the maximality of $G$.
This means that every vertex in $I^f_{3^+}$ has at least three neighbors in $\{u_1,u_2,\ldots,u_{d(f)}\}$.

We further define $B$ as the plane graph obtained from $A$ by (1) replacing every triple $(u,x,v)$ where $x\in L_2$, $u,v\in S$ and $uxv$ forms a path by a single edge, (2) removing vertices of $V(CL_2)$, (3) contracting every chain from $C_{3^+}$ into a single vertex, and (4) contracting every chain from $C_{2^-}$ into a single edge. 
By (4) we mean that every maximal chain $d=x_1,\ldots,x_i$ of $I_2$ vertices which is contained in a chain from $C_{2^-}$, is replaced by the edge $yz$ where $y$ and $z$ are the forest neighbors (in $L_{3^+}\cup I_{3^+}$) of $x_1$ and $x_i$ outside the chain $d$.
Let us call the vertices of $B$ that are not on the boundary of $f$ as {\em inner vertices}.

Note that the set of inner vertices is in a bijection with $L^f_{3^+} \cup I^f_{3^+} \cup C^f_{3^+}$.
Moreover, $I$ forms a tree, since $F^f$ is a tree.
Also, each inner vertex has at least three neighbors in $\{u_1,u_2,\ldots,u_{d(f)}\}$.
We show that $|I|\le d(f)-2$ by the induction on $d(f)$.
When $d(f)=2$ the claim follows since each inner vertex has at least three neighbors on the boundary of $f$.
Now assume $d(f)>2$. 
Let $x$ be leaf in the tree $I$. 
Then the edges from $x$ to the boundary of face $f$ split $F$ into at least three different faces.
The subtree $I-x$ lies in one of these faces, say face bounded by the cycle $xu_iu_{i+1}\cdots u_j x$.
We remove $x$ and vertices $u_{j+1},\ldots,u_{i-1}$ (there is at least one of them) and we add edge $u_iu_j$.
The outer face of the resulting graph is of length at most $d(f)-1$, so we can apply induction and the claim follows.

\end{proof}

\begin{lemma}\label{lem:fdf2}
For every face $f$ in $H_S$ of length at least three, \[|V_F^f\setminus V(CL^f_2)|\leq \ell\cdot(d(f)-2) - (\ell -1).\]
\end{lemma}
\begin{proof}
We have 
\[|V_F^f\setminus V(CL^f_2)|\leq  |L^f_{3^+}|+|I^f_{3^+}|+|V(C^f_{3^+})|+|V(C^f_{2^-})|.\]

By Lemma~\ref{lem:construction}$(v)$ we get
\begin{equation}
\label{eq:bound1}
|V_F^f\setminus V(CL^f_2)|\leq  |L^f_{3^+}|+|I^f_{3^+}|+\ell|C^f_{3^+}|+(\ell-1)|C^f_{2^-}|.
\end{equation}
By Lemma~\ref{lem:construction}$(iv)$, $|C^f_{2^-}|$ is bounded by the number of edges of the inner forest $F_I$.
Hence, $|C^f_{2^-}|\leq |L^f_{3^+}|+|I^f_{3^+}|-1$ when $|L^f_{3^+}|+|I^f_{3^+}|> 0$ and $|C^f_{2^-}|=0$ otherwise.
In the prior case, by~\eqref{eq:bound1} we get that
\[|V_F^f\setminus V(CL^f_2)|\leq  \ell (|L^f_{3^+}|+|I^f_{3^+}|+|C^f_{3^+}|)-(\ell-1),\] and the result then follows from Lemma~\ref{cl:df2}. 
Hence it suffices to prove the claim when $|L^f_{3^+}|=|I^f_{3^+}|=|C^f_{2^-}|=0$.
Then the forest $F^f$ is a non-empty collection of paths, each with both endpoints in $L_2$. 
Let $c$ be such a path on $p$ vertices $x_1,\ldots,x_p$. 
Then $x_1\in L_2$ and $x_1$ has exactly two neighbors $u,v$ in $S$.
Let $i$ be the largest such that $N(\{x_1,\ldots,x_i\})\cap S = \{u,v\}$.
By definition, $(x_1,\ldots,x_i)$ is a chain in $CL^f_2$.
We infer that if $i=p$ for every such path, then $|V_F^f\setminus V(CL^f_2)| = 0$ and the claim follows.
Hence we can assume that $i<p$, i.e., $x_{i+1}$ has a neighbor in $S\setminus\{u,v\}$.
Then, by definition, $(x_1,\ldots,x_{i+1})$ is a chain in $C^f_{3^+}$. 
Since $(x_1,\ldots,x_i)\in CL^f_2$, we get $|\{x_1,\ldots,x_{i+1}\}\setminus V(CL^f_2)| = 1$. 
Hence, 
\[|V_F^f\setminus V(CL^f_2)|\leq  1+\ell(|C^f_{3^+}|-1),\]
what, by Lemma~\ref{cl:df2} , is bounded by $1+\ell\cdot(d(f)-3) = \ell\cdot(d(f)-2) - (\ell -1)$, as  required.

\end{proof}

\begin{lemma}\label{lem:face-len-2}
For every face $f$ in $H_S$ of length two, $V_F^f \subseteq  V(CL^f_2)$.
\end{lemma}

\begin{proof}
Since the boundary of $f$ has only two solution vertices, $F^f$ contains no vertices of $L_{3^+}^f$, $V(C_{3^+})^f$ or $I_{3^+}^f$. 
Then by Lemma~\ref{lem:construction}$(iv)$, $C_{2^-}^f$ is also empty. The claim follows.
\end{proof}


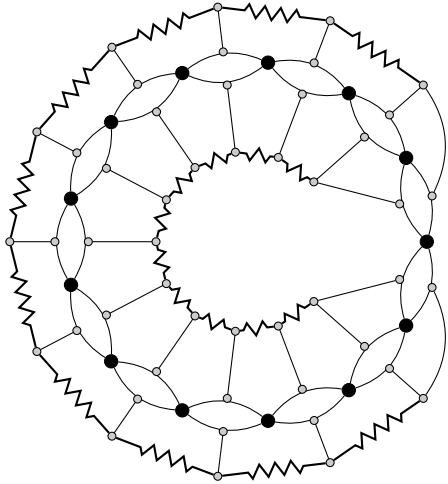
\begin{figure}[t]
\centering
\begin{tikzpicture}[scale=1.2,auto]
\tikzstyle{blacknode}=[draw,circle,fill=black,minimum size=5pt,inner sep=0pt]
\tikzstyle{greennode}=[draw,circle,fill=black!20!white,minimum size=3pt,inner sep=0pt]
\tikzstyle{greennodeb}=[draw,circle,fill=black!20!white,minimum size=3pt,inner sep=0pt]
\tikzstyle{spring}=[thick,decorate,decoration={zigzag,pre length=0.05cm,post length=0.05cm,segment length=6}];
\tikzstyle{springb}=[thick,decorate,decoration={zigzag,pre length=0.3cm,post length=0.3cm,segment length=6}];
\foreach \i in {0,...,12}
{\draw (0,0)
++(360*\i/13:2) node[blacknode] (a\i) {};}

\foreach \i in {0,...,12}
{\pgfmathtruncatemacro{\j}{mod(\i+1,13)} 
\draw (a\i) edge [bend right] node {} (a\j);
\draw (a\i) edge [bend left] node {} (a\j);
    }

\foreach \i in {2,...,12}
{\pgfmathtruncatemacro{\j}{mod(\i-1,13)} 
\draw (0,0)
++(360*\j/13+180/13:1) node[greennode] (b\i) {}
--++(360*\j/13+180/13:0.75) node[greennode] (c\i) {};
}

\foreach \i in {0,1}
{\pgfmathtruncatemacro{\j}{mod(\i-1,13)} 
\draw (0,0)
++(360*\j/13+180/13:1.75) node[greennode] (c\i) {};}

\foreach \i in {3,...,12}
{\pgfmathtruncatemacro{\j}{mod(\i-1,13)} 
\draw[spring] (b\j) -- (b\i);
    }
    
\draw (b12) -- (c0);    
\draw (b2) -- (c1);

\foreach \i in {2,...,12}
{\pgfmathtruncatemacro{\j}{mod(\i-1,13)} 
\draw (0,0)
++(360*\j/13+180/13:2.12) node[greennode] (e\i) {}
--++(360*\j/13+180/13:0.5) node[greennode] (f\i) {};
}

\foreach \i in {0,1}
{\pgfmathtruncatemacro{\j}{mod(\i-1,13)} 
\draw (0,0)
++(360*\j/13+180/13:2.12) node[greennode] (e\i) {};}

\draw (f12) edge [bend right] node {} (e0);    
\draw (f2) edge [bend left] node {} (e1);

\foreach \i in {3,...,12}
{\pgfmathtruncatemacro{\j}{mod(\i-1,13)} 
\draw[springb] (f\i) -- (f\j);
    }
\end{tikzpicture}
\caption{A tight example. The big black vertices are solution vertices, the small gray ones are forest vertices. 
The zigzag edges represent paths of $\ell-1$ forest vertices, each adjacent to the two available solution vertices. Asymptotically for larger cycles, we have $2\ell + 3$ forest vertices for each solution vertex.}
\label{fig:tight}
\end{figure}

Now we proceed to the bound of Theorem~\ref{thm:size-bound}.
By Lemmas~\ref{lem:fdf2} and~\ref{lem:face-len-2} we have 

\[|V_F| \leq |V(CL_2)|+\sum_{f \in F_{S,3+}} (\ell(d(f)-2)-(\ell-1)) \]

By Lemma~\ref{lem:L23E} we get

\begin{equation*}
 \begin{split}
   |V_F|  & \le 2(|E_S| - |F_{S,2}|) + \sum_{f \in F_{S,3+}} \left(\ell(d(f)-2)-(\ell-1)\right) \\
        & =   2(|E_S| - |F_{S,2}|) + \sum_{f \in F_S} \left(\ell(d(f)-2)-(\ell-1)\right) + (\ell-1)|F_{2,S}|\\
        & =   (2\ell+2)|E_S| - (3\ell-1)|F_S| + (\ell-3)|F_{2,S}|\\
        & =   (2\ell+2)|E_S| - (2\ell+2)|F_S| - (\ell-3)|F_S| + (\ell-3)|F_{2,S}| \\
        & \le (2\ell+2)(|E_S| - |F_S|).
 \end{split}
\end{equation*}
By Lemma~\ref{cl:solutionconnected} graph $H_S$ is connected, so we can apply Euler's formula $|S|-|E_S|+|F_S|=2$. 
Thus,
\begin{equation*}
 \begin{split}
   |V(G)| = |V_F|+|S|  & \le  (2\ell+2)(|S|-2) + |S|, \\
        & =    (2\ell+3)k - (4\ell+4).
 \end{split}
\end{equation*}

This concludes the proof of Theorem~\ref{thm:size-bound}.
By Lemma~\ref{lem:indpath5}, we can put $\ell=5$, which results in $|V(G)|\leq 13k-24$.
In Figure \ref{fig:tight} we show an example of a graph, where our reduction rules do not apply and our analysis is tight (up to a constant additive term).

\section{Reducing induced $5$-paths with at most two neighbors}
\label{sec:5path}

This section is devoted to a proof of Lemma~\ref{lem:indpath5}. Let us recall its statement here. 

\newcounter{savecnt}
\setcounter{savecnt}{\value{theorem}}
\setcounter{theorem}{\value{indpath5}}
\begin{lemma}[restated]
 Assume there is an induced path $u_0ux_1x_2x_3vv_0$ such that for some vertices $w_1$, $w_2$ outside the path $N(\{u,x_1,x_2,x_3,v\})\setminus\{u_0,v_0\} = \{w_1,w_2\}$.
 Then there is an instance $(G',k')$ with $|V(G')|<|V(G)|$ such that $(G,k)$ is a yes-instance iff $(G',k')$ is a yes-instance and $k'\le k$.
\end{lemma}
\setcounter{theorem}{\value{savecnt}}

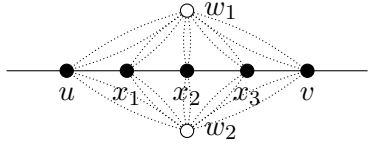
\begin{figure}[t]
\centering
\begin{tikzpicture}[scale=\scalefactor]
\tikzstyle{whitenode}=[draw,circle,fill=white,minimum size=5pt,inner sep=0pt]
\tikzstyle{blacknode}=[draw,circle,fill=black,minimum size=5pt,inner sep=0pt]
\tikzstyle{texte}=[circle,minimum size=5pt,inner sep=0pt]
\tikzstyle{innerWhite} = [semithick, white,line width=0.5pt]
\draw (-3,0) node (u0) {}
-- ++(1,0) node[blacknode] [label=-90:$u$] (u) {}
-- ++(1,0) node[blacknode] [label=-90:$x_1$] (x1) {}
-- ++(1,0) node[blacknode] [label=-90:$x_2$] (x2) {}
-- ++(1,0) node[blacknode] [label=-90:$x_3$] (x3) {}
-- ++(1,0) node[blacknode] [label=-90:$v$] (v) {}
-- ++(1,0) node (v0) {};

\draw (0,1) node[whitenode] [label=right:$w_1$] (w1) {}; 
\draw (0,-1) node[whitenode] [label=right:$w_2$] (w2) {}; 

\foreach \i in {u,x1,x2,x3,v}
{
\draw (w1) edge [bend left=10,densely dotted] node {} (\i);
\draw (w1) edge [bend right=10,densely dotted] node {} (\i);
\draw (w2) edge [bend left=10,densely dotted] node {} (\i);
\draw (w2) edge [bend right=10,densely dotted] node {} (\i);
}


\end{tikzpicture}
\caption{Configuration from Lemma~\ref{lem:indpath5}.}
\label{fig:5path}
\end{figure}

Denote the path $ux_1x_2x_3v$ by $P$.
Let $Q=G[V(P)\cup\{w_1,w_2\}]$.
By symmetry we assume $|N(w_1)\cap V(P)|\ge |N(w_2)\cap V(P)|$.
Then also $|N(w_1)\cap V(P)|\ge 3$, for otherwise Rule~\ref{r:deg2} applies.
In our proof of Lemma~\ref{lem:indpath5} we do not apply a single rule, but one of four rules.
The kernelization algorithm finds the family $\Ss$ of all feedback sets of $Q$. 
(Note that there is a bounded number of such sets.)
Based on the structure of $\Ss$, one of the four rules is chosen and applied.
Let us also define $\delta(Q) = N_G(X) = \{u,v,w_1,w_2\}$.
Let us state the rules now.

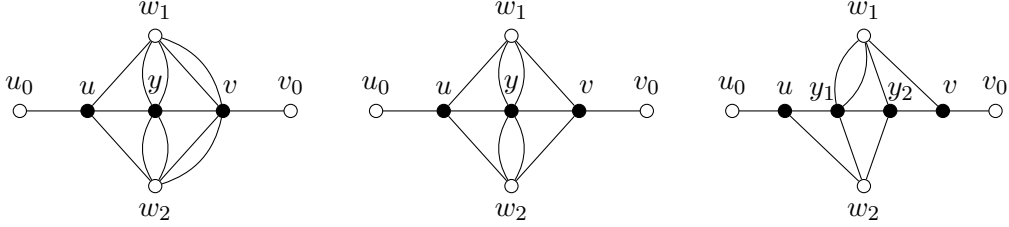
\begin{figure}[t]
\centering
\begin{tikzpicture}[scale=1]
\tikzstyle{whitenode}=[draw,circle,fill=white,minimum size=5pt,inner sep=0pt]
\tikzstyle{blacknode}=[draw,circle,fill=black,minimum size=5pt,inner sep=0pt]
\draw (0,0) node[whitenode] [label=90:$u_0$] (u) {}
-- ++(0:0.9cm) node[blacknode] [label=90:$u$] (y1) {}
-- ++(0:0.9cm) node[blacknode] (y2) {}
-- ++(0:0.9cm) node[blacknode] (y3) {}
-- ++(0:0.9cm) node[whitenode] [label=90:$v_0$] (v) {};

\draw (y2) node[above=1mm] {$y$};
\draw ($(y3)+(1mm,0)$) node[above=1mm] {$v$};

\draw (1.8,1) node[whitenode] [label=90:$w_1$] (w1) {};
\draw (1.8,-1) node[whitenode] [label=-90:$w_2$] (w2) {};

\draw (y1) edge node {} (w1); 
\draw (y1) edge node {} (w2); 

\draw (y2) edge [bend left] node {} (w1); 
\draw (y2) edge [bend right] node {} (w1);
\draw (y2) edge [bend left] node {} (w2); 
\draw (y2) edge [bend right] node {} (w2);

\draw (y3) edge node {} (w1); 
\draw (y3) edge [bend right] node {} (w1);
\draw (y3) edge [bend left] node {} (w2); 
\draw (y3) edge node {} (w2);

\end{tikzpicture}
\quad
\begin{tikzpicture}[scale=1]
\tikzstyle{whitenode}=[draw,circle,fill=white,minimum size=5pt,inner sep=0pt]
\tikzstyle{blacknode}=[draw,circle,fill=black,minimum size=5pt,inner sep=0pt]
\draw (0,0) node[whitenode] [label=90:$u_0$] (u) {}
-- ++(0:0.9cm) node[blacknode] [label=90:$u$] (y1) {}
-- ++(0:0.9cm) node[blacknode] (y2) {}
-- ++(0:0.9cm) node[blacknode] (y3) {}
-- ++(0:0.9cm) node[whitenode] [label=90:$v_0$] (v) {};

\draw (y2) node[above=1mm] {$y$};
\draw ($(y3)+(1mm,0)$) node[above=1mm] {$v$};

\draw (1.8,1) node[whitenode] [label=90:$w_1$] (w1) {};
\draw (1.8,-1) node[whitenode] [label=-90:$w_2$] (w2) {};

\draw (y1) edge node {} (w1); 
\draw (y1) edge node {} (w2); 
\draw (y2) edge [bend left] node {} (w1); 
\draw (y2) edge [bend right] node {} (w1);
\draw (y2) edge [bend left] node {} (w2); 
\draw (y2) edge [bend right] node {} (w2);
\draw (y3) edge node {} (w1);   
\draw (y3) edge node {} (w2);  
\end{tikzpicture}
\quad
\begin{tikzpicture}[scale=1]
\tikzstyle{whitenode}=[draw,circle,fill=white,minimum size=5pt,inner sep=0pt]
\tikzstyle{blacknode}=[draw,circle,fill=black,minimum size=5pt,inner sep=0pt]
\draw (0,0) node[whitenode] [label=90:$u_0$] (u) {}
-- ++(0:0.7cm) node[blacknode] [label=90:$u$] (y1) {}
-- ++(0:0.7cm) node[blacknode] (y11) {}
-- ++(0:0.7cm) node[blacknode] (y2) {}
-- ++(0:0.7cm) node[blacknode] (y3) {}
-- ++(0:0.7cm) node[whitenode] [label=90:$v_0$] (v) {};

\draw ($(y2)+(1.4mm,0)$) node[above] {$y_2$};
\draw ($(y11)+(-2mm,0)$) node[above] {$y_1$};
\draw ($(y3)+(1mm,0)$) node[above=1mm] {$v$};


\draw ($(y11)!0.5!(y2)+(0,1)$) node[whitenode] [label=90:$w_1$] (w1) {};
\draw ($(y11)!0.5!(y2)-(0,1)$) node[whitenode] [label=-90:$w_2$] (w2) {};

\draw (y1) edge node {} (w2); 
\draw (y11) edge [bend left] node {} (w1); 
\draw (y11) edge [bend right] node {} (w1);
\draw (y11) edge node {} (w2);
\draw (y3) edge node {} (w1);   
\draw (y2) edge node {} (w1);  
\draw (y2) edge node {} (w2);  
\end{tikzpicture}
\caption{Gadgets of rules~\ref{r:gadget:u},~\ref{r:gadget:uv} and \ref{r:gadget:4}.}
\label{fig:gadgets}
\end{figure}

\rrule{r:del-w1}{Assume every set $S_Q\in\mathcal{S}$ satisfies at least one of the conditions below:
\begin{enumerate}[$(1)$]
 \item $|S_Q|\ge 3$,
 \item $Q-S_Q$ contains a $(u,v)$-path,
 \item $w_1\in S_Q$, or
 \item $w_2\in S_Q$.
\end{enumerate}
Then remove vertex $w_1$ and decrease $k$ by one.}

\rrule{r:n(w2)=2}{If $|N(w_2)\cap V(P)|\le 2$ Then remove vertex $w_1$ and decrease $k$ by one.}

\rrule{r:gadget:u}{If every set $S_Q\in\mathcal{S}$ satisfies at least one of the conditions (1)-(4) or
\begin{enumerate}[$(1)$]
 \setcounter{enumi}{4}
 \item \label{cond:v-u,w1,w2}the sets $\{v\}$ and $\{u,w_1,w_2\}$ are the equivalence classes of $R_{Q-S_Q,\delta(Q)}$,
\end{enumerate}
then replace $Q$ by the gadget from Figure~\ref{fig:gadgets} (left), i.e., remove $x_1,x_2,x_3$ and edges in $G[Q]$, and add a vertex $y$ and edges $uy$, $yv$, $uw_1$, $uw_2$ and double edges $vw_1$, $vw_2$, $yw_1$, $yw_2$.}

\rrule{r:gadget:uv}{If every set $S_Q\in\mathcal{S}$ satisfies at least one of the conditions (1)-(5) or
\begin{enumerate}[$(1)$]
 \setcounter{enumi}{5}
 \item the sets $\{u\}$ and $\{v,w_1,w_2\}$ are the equivalence classes of $R_{Q-S_Q,\delta(Q)}$,
\end{enumerate}
then replace $Q$ by the gadget from Figure~\ref{fig:gadgets} (middle), i.e., remove $x_1,x_2,x_3$ and edges in $G[Q]$, and add a vertex $y$ and edges $uy$, $yv$, $uw_1$, $uw_2$, $vw_1$, $vw_2$ and double edges $yw_1$, $yw_2$.}

\rrule{r:gadget:4}{If every set $S_Q\in\mathcal{S}$ satisfies at least one of the conditions (1)-(5) or
\begin{enumerate}[$(1)$]
 \setcounter{enumi}{6}
 \item the sets $\{u,w_2\}$ and $\{v,w_1\}$ are the equivalence classes of $R_{Q-S_Q,\delta(Q)}$,
\end{enumerate}
then replace $Q$ by the gadget from Figure~\ref{fig:gadgets} (right), i.e., remove $x_1,x_2,x_3$ and edges in $G[Q]$, and add vertices $y_1,y_2$ and edges $uy_1$, $y_1y_2$, $y_2v$, $uw_2$, $y_1w_2$, $y_2w_2$, $y_2w_1$, $vw_1$ and a double edge $y_1w_1$.}

Note that the graph modifications in rules~\ref{r:gadget:u}--\ref{r:gadget:4} can be phrased as gadget replacements $(X,Y,E_I)$ in graph $G$, where $X=\{x_1,x_2,x_3\}$, $Y=\{y\}$ for rules~\ref{r:gadget:u} and \ref{r:gadget:uv}, while $Y=\{y_1,y_2\}$ for Rule~\ref{r:gadget:4}.

\begin{lemma}
\label{lem:completeness}
If there is an induced path described in Lemma~\ref{lem:indpath5}, then one of rules~\ref{r:del-w1}--\ref{r:gadget:4} applies.
Moreover, if Rule~\ref{r:gadget:4} is applied, then $\Ss$ contains both a set satisfying $(5)$ and a set satisfying $(7)$.
\end{lemma}

\begin{proof}
Assume rules~\ref{r:del-w1}--\ref{r:gadget:uv} do not apply.
We will show that Rule~\ref{r:gadget:4} applies.
Then there is a set $S_Q\in\Ss$ that satisfies none of (1)--(6).
Since (1), (3) and (4) do not hold for $S_Q$, we infer that $|S_Q|\le 2$ and $S_Q\subseteq\{u,x_1,x_2,x_3,v\}$.
Recall that $|N(w_1)\cap V(P)|\ge 3$.
Let $p,q,r$ be three arbitrary vertices of $N(w_1)\cap V(P)$, sorted by increasing distance from $u$ in the path $ux_1x_2x_3v$.
Assume $|S_Q|\le 1$. Then $q\in S_Q$, for otherwise $S_Q$ is not a feedback vertex set of $Q$.
But then there is a $(u,v)$-path in $Q-S_Q$: from $u$ follow the path $ux_1x_2x_3v$ to $p$, then via $w_1$ to $r$ and again follow the path $ux_1x_2x_3v$ to $v$. Hence $S_Q$ satisfies (2), a contradiction. In what follows we assume $|S_Q|=2$. 

\mycase{1:} In $Q-S_Q$ there is an edge $ab\in E(P)$.
Then in particular $a,b\not\in S_Q$.
Since Rule~\ref{r:deg2} does not apply, both $a$ and $b$ have a neighbor in $\{w_1,w_2\}$.
However, since $Q-S_Q$ does not contain cycles, the neighbor of $a$ is different than the neighbor of $b$.
Hence, $Q-S_Q$ contains path $w_1abw_2$ or path $w_1baw_2$.
If $u\in\{a,b\}$, $S_Q$ satisfies (5), a contradiction, and if $v\in\{a,b\}$, $S_Q$ satisfies (6), a contradiction.
Hence $\{u,v\}\cap\{a,b\}=\emptyset$.
It follows that $\{a,b\}=\{x_1,x_2\}$ or $\{a,b\}=\{x_2,x_3\}$.
In the former case $x_3 \in S_Q$ and in the latter case $x_1 \in S_Q$, for otherwise $Q-S_Q$ has a cycle because $x_3$ (resp. $x_1$) has a neighbor in $\{w_1,w_2\}$ by Rule~\ref{r:deg2}.
Since $|S_Q|= 2$, it follows that exactly one of $u$, $v$ is not in $S_Q$. 
But by Rule~\ref{r:deg2} both of them have a neighbor in $\{w_1,w_2\}$, so (5) or (6) is satisfied, a contradiction.

\mycase{2:} $Q-S_Q$ contains no edge of $E(P)$.
Then $S_Q=\{x_1,x_3\}$.

If $uw_2\in E$, then $vw_2\not\in E$, for otherwise (2) holds.
It follows that $vw_1\in E$, for otherwise Rule~\ref{r:deg2} applies.
Then $uw_1\not\in E$, for otherwise (2) holds.

By the same arguments, if $uw_2\not\in E$, then $uw_1\in E$, $vw_1\not\in E$ and $vw_2\in E$.
We see that the two cases above are symmetric, so let us consider only the former one, i.e., $uw_2,vw_1\in E$ and $uw_1,vw_2\not\in E$.
Then also $w_1w_2\not\in E$, for otherwise (2) holds.

\mycase{2.1:} $x_2w_1\not\in E$.
Since Rule~\ref{r:deg2} does not apply, $x_2w_2\in E$.
Since $|N(w_1)\cap V(P)|\ge 3$, $x_1,x_3\in N(w_1)$.
Hence $N_Q(w_1)=\{x_1,x_3,v\}$.
Since $|N(w_2)\cap V(P)|\le |N(w_1)\cap V(P)|$ and $|N(w_2)\cap V(P)|\ge 3$ as Rule~\ref{r:n(w2)=2} does not apply, we get $|N(w_2)\cap V(P)|=3$.
It follows that either $N_Q(w_2)=\{u,x_1,x_2\}$ or $N_Q(w_2)=\{u,x_2,x_3\}$.
In both cases $S_Q$ satisfies (7).

\mycase{2.1.1:} $N_Q(w_2)=\{u,x_1,x_2\}$. 
Consider an arbitrary $S'_Q\in \Ss$ that satisfies none of (1)--(4) or (7).
Similarly as we argued for $S_Q$, $S'_Q\subseteq\{u,x_1,x_2,x_3,v\}$ and $|S'_Q|\le 2$.
Then $x_1\in S'_Q$, for otherwise it is impossible to hit all three cycles $w_2ux_1$, $w_2x_1x_2$ and $w_1x_3v$. 
It follows that $S'_Q$ contains $x_3$ or $v$, but in the former case (7) holds. 
Hence $S'_Q=\{x_1,v\}$, and (5) holds for $S'_Q$.
This proves our claim.

\mycase{2.1.2:} $N_Q(w_2)=\{u,x_2,x_3\}$. Then we rename vertices of $Q$: swap names of $w_1$ and $w_2$, $u$ and $v$, $x_1$ and $x_3$, obtaining $N_Q(w_1)=\{x_1,x_2,v\}$ and $N_Q(w_2)=\{u,x_1,x_3\}$.
Notice that in this new setting $\{x_1,x_3\}$ is still a feedback vertex set of $Q$ and satisfies (7). 
In the new setting, consider an arbitrary $S'_Q\in \Ss$ that satisfies none of (1)--(4) or (7).
Similarly as we argued for $S_Q$, $S'_Q\subseteq\{u,x_1,x_2,x_3,v\}$ and $|S'_Q|\le 2$.
Then $x_1\in S'_Q$, for otherwise it is impossible to hit all three cycles $w_2ux_1$, $w_1x_1x_2$ and $w_1x_1w_1x_3v$. 
It follows that $S'_Q$ contains $x_2$, $x_3$ or $v$, but in the first case (2) holds and in the second case (7) holds. 
Hence $S'_Q=\{x_1,v\}$, and (5) holds for $S'_Q$.
This proves our claim.

\mycase{2.2:} $x_2w_1\in E$.
Then $x_2w_2\not\in E$, for otherwise (2) holds.
Since $|N(w_2)\cap V(P)|\ge 3$ by Rule~\ref{r:n(w2)=2}, $x_1,x_3\in N(w_2)$.
Hence $N_Q(w_2)=\{u,x_1,x_3\}$.
Since $|N(w_1)\cap V(P)|\ge 3$, we have $x_1w_1\in E$ or $x_3w_1\in E$.
If both edges $x_1w_1$ and $x_3w_1$ are present, then Rule~\ref{r:gamma} applies, a contradiction.
It follows that either $N_Q(w_2)=\{x_1,x_2,v\}$ or $N_Q(w_2)=\{x_2,x_3,v\}$.
The former case was already considered in Case 2.1.2 (after renaming vertices).
In the latter case, we rename vertices by swapping names of $w_1$ and $w_2$, $u$ and $v$, $x_1$ and $x_3$.
Thus we obtain the already considered Case 2.1.1.
\end{proof} 

The lemma below will be very useful in proving that the rules above are correct in particular settings of subgraph $Q$.

\begin{lemma}
\label{lem:common}
  In each of the situations below all rules~\ref{r:del-w1}--\ref{r:gadget:4} are correct.
  \begin{enumerate}[$(i)$]
   \item $(G,k)$ is a no-instance.
   \item there is a solution $S$ to the instance $(G,k)$ such that $|V(Q)\cap S| \ge 2$ and
   \begin{enumerate}[$({ii}.1)$]
    \item $|V(Q)\cap S| \ge 3$, or
    \item in $Q-S$ there is an $(x_1,x_5)$ path, or
    \item $w_1 \in S$, or
    \item $w_2 \in S$.
   \end{enumerate}
\end{enumerate}
\end{lemma}

\begin{proof}
 We begin with $(i)$. We need to show that for each of the rules~\ref{r:del-w1}--\ref{r:gadget:4} the resulting instance $(G',k')$ is a no-instance. Assume the contrary, i.e., let $S'$ be a solution of $(G',k')$. 
 We will show that there is a solution $S$ of $(G,k)$, contradicting our assumption.
 For rules~\ref{r:del-w1} and~\ref{r:n(w2)=2} we see that $S=S'\cup\{w_1\}$ works.
 Now focus on the remaining rules.
 Let $Y=\{y\}$ for rules~\ref{r:gadget:u} and~\ref{r:gadget:uv}, and $Y=\{y_1,y_2\}$ for Rule~\ref{r:gadget:4}.
 Let $Q'=G'[N[Y]]$.
 It is easy to verify that for each of the three rules $|V(Q')\cap S| \ge 2$.
 If $|V(Q')\cap S| \ge 3$ then we see that $S = S'\setminus V(Q') \cup \{w_1,w_2,x_1\}$ works.
 Hence we are left with the case $|V(Q')\cap S| = 2$.
 
 Assume there is a $(u,v)$-path in $G'-S'$.
 Then we put $S = S'\setminus V(Q') \cup \{w_1,w_2\}$.
 Note that the equivalence classes of $R_{Q-S,\delta(Q)}$ are $\{u,v\}$, $\{w_1\}$ and $\{w_2\}$.
 Hence, by Lemma~\ref{lem:fvs-transfer} (for $A=G'$ and $B=G$) $S$ is a feedback vertex set of $G$, so $(G,k)$ is a yes-instance, a contradiction.
 
 Hence we can assume that there is no $(u,v)$ path in $G'-S'$.
 Note that this implies that $V(Q')\cap S'$ is equal to $\{u,y\}$ for Rule~\ref{r:gadget:u}, $\{u,y\}$ or $\{y,v\}$ for Rule~\ref{r:gadget:uv}, and $\{y_1,y_2\}$ or $\{y_1,v\}$ for Rule~\ref{r:gadget:4}.
 
 Consider Rule~\ref{r:gadget:u}. Since Rule~\ref{r:del-w1} does not apply, there is at least one feedback vertex set $S_Q\in\Ss$ of size at most two which satisfies~(\ref{cond:v-u,w1,w2}). Since $V(Q')\cap S'=\{u,y\}$, the equivalence classes of $R_{Q'-S',\delta(Q)}$ are $\{v\}$ and $\{u,w_1,w_2\}$, hence by Lemma~\ref{lem:fvs-transfer} (for $A=G'$ and $B=G$) $S=S'\setminus V(Q') \cup S_Q$ is a feedback vertex set $G$, so $(G,k)$ is a yes-instance.
 
 Now consider Rule~\ref{r:gadget:uv}. 
 Note that it cannot happen that $\Ss$ contains only sets that satisfy (1)-(4), or (6), because then Rule~\ref{r:gadget:u} applies to $Q$ with vertices renamed (swap the names of $u$ and $v$, $x_1$ and $x_3$).
 Since rules~\ref{r:del-w1} and~\ref{r:gadget:u} do not apply to $Q$, $\Ss$ contains both a feedback vertex set $S_Q^1\in\Ss$ which satisfies (5) and $S_Q^2\in\Ss$ which satisfies (6), and $|S_Q^{(5)}|,|S_Q^{(6)}|\le 2$.
 Since  $V(Q')\cap S'=\{u,y\}$ or $V(Q')\cap S'=\{y,v\}$ the equivalence classes of $R_{Q'-S',\delta(Q)}$ are $\{u\}$ and $\{v,w_1,w_2\}$, or $\{v\}$ and $\{u,w_1,w_2\}$, respectively. 
 In the former case we put $S=S'\setminus V(Q') \cup S_Q^{(6)}$ and in the latter one  $S=S'\setminus V(Q') \cup S_Q^{(5)}$.
 By Lemma~\ref{lem:fvs-transfer} (for $A=G'$ and $B=G$) $S$ is a feedback vertex set $G$, so $(G,k)$ is a yes-instance.
 
 Finally, consider Rule~\ref{r:gadget:4}.
 By Lemma~\ref{lem:completeness}, $\Ss$ contains both a set $|S_Q^{(5)}|$ satisfying (5) and a set $|S_Q^{(7)}|$ satisfying (7).
 Since  $V(Q')\cap S'=\{y_1,y_2\}$ or $V(Q')\cap S'=\{y_1,v\}$ the equivalence classes of $R_{Q'-S',\delta(Q)}$ are $\{u,w_2\}$ and $\{v,w_1\}$, or $\{v\}$ and $\{u,w_1,w_2\}$, respectively. 
 In the former case we put $S=S'\setminus V(Q') \cup S_Q^{(7)}$ and in the latter one  $S=S'\setminus V(Q') \cup S_Q^{(5)}$.
 By Lemma~\ref{lem:fvs-transfer} (for $A=G'$ and $B=G$) $S$ is a feedback vertex set $G$, so $(G,k)$ is a yes-instance.
 This ends the proof of $(i)$.
 
 We proceed to $(ii)$. We need to show that in each of the cases $(ii.1)$--$(ii.4)$, $(G',k')$ is a yes-instance.
 To this end we will show a feedback vertex set $S'$ of size at most $k'$ in $G'$.
 
 For $(ii.1)$ we pick $S'=(S\setminus V(Q)) \cup \{u,w_1,w_2\}$. 
 Then the equivalence classes of $R_{Q'-S',\delta(Q)}$ are all singletons, so by Lemma~\ref{lem:fvs-transfer} (for $A=G$, $B=G'$) $S'$ is a feedback vertex set of $G'$.
 
 For $(ii.2)$ we pick $S'=(S\setminus V(Q)) \cup \{w_1,w_2\}$. 
 Then the equivalence classes of $R_{Q'-S',\delta(Q)}$ are $\{u,v\}$, $\{w_1\}$ and $\{w_2\}$, so by Lemma~\ref{lem:fvs-transfer} (for $A=G$, $B=G'$) $S'$ is a feedback vertex set of size at most $k$ in $G'$.
 
 For $(ii.3)$ we consider two cases.
 In case of rules~\ref{r:del-w1} and~\ref{r:n(w2)=2} it is clear that $S'=S\setminus \{w_1\}$ works. 
 Hence we can assume that Rule~\ref{r:n(w2)=2} does not apply.
 Let $p,q,r$ be arbitrary three vertices of $N(w_2)\cap V(P)$, in the order of increasing distance from $u$ in the path $P$.
 We can assume that $|V(Q)\cap S|\le 2$, for otherwise we use $(ii.1)$.
 Then $r\in S$, for otherwise we need to include both $p$ and $q$ (at least) to hit all the cycles of $Q$.
 Hence $V(Q)\cap S = \{w_1,r\}$. 
 But then there is a path going from $u$ to $p$ along $P$, then via $w_2$ to $r$ and to $v$ along $P$. 
 Hence we apply $(ii.2)$.
 
 For $(ii.4)$, we can assume that $|V(Q)\cap S|\le 2$ and $w_1\not\in S$, for otherwise we apply $(ii.1)$ or $(ii.3)$.
 Since $|N(w_1)\cap V(P)|\ge 3$ we can pick three vertices $p,q,r$ of $N(w_1)\cap V(P)$, in the order of increasing distance from $u$ in the path $P$.
 Then $r\in S$, for otherwise we need to include both $p$ and $q$ (at least ) to hit all cycles of $Q$.
 Hence $V(Q)\cap S = \{w_2,r\}$. 
 But then there is a path going from $u$ to $p$ along $P$, then via $w_1$ to $r$ and to $v$ along $P$. 
 Hence we apply $(ii.2)$.
\end{proof}

\begin{lemma}
 Rules~\ref{r:del-w1}--\ref{r:gadget:4} are correct.
\end{lemma}

\begin{proof}
 By Lemma~\ref{lem:common}$(i)$ it suffices to prove the correctness when there is a feedback vertex set $S$ of size at most $k$.

 Since $|N(w_1)\cap V(P)|\ge 3$, we see that $Q$ contains cycles and hence $|V(Q)\cap S|> 0$, for otherwise $S$ is not a feedback vertex set.
 
 First we consider the case $|V(Q)\cap S|=1$.
 Assume $|N(w_2)\cap V(P)|\ge 3$.
 By Rule~\ref{r:deg2}, $Q$ has a subgraph $R$ consisting of path $P$, three edges, each joining a different vertex of $P$ with $w_1$, and another three edges, each joining a different vertex of $P$ with $w_2$.
 Note that $R$ has $7$ vertices, $10$ edges and maximum degree $4$.
 Then $R-S$ has $6$ vertices and at least $6$ edges, so $S$ is not a feedback vertex set, a contradiction.
 Assume $|N(w_2)\cap V(P)|\le 2$.
 Then Rule~\ref{r:del-w1} or Rule~\ref{r:n(w2)=2} applies.
 We will show that $V(Q)\cap S=\{w_1\}$; then clearly $S'=S\setminus\{w_1\}$ is a solution of $(G',k-1)$.
 Build a subgraph $R$ of $Q$ as follows.
 Start with $R=P$.
 For each $x\in N(w_2)\cap V(P)$, add a single edge $xw_2$ to $R$.
 Next, for each $x\in V(P)\setminus N(w_2)$, add a single edge $xw_1$ to $R$ (which exists since Rule~\ref{r:deg2} does not apply).
 Note that $R$ has $7$ vertices, $9$ edges and every vertex of $R$ apart from $w_1$ has degree at most three in $R$.
 It follows that $V(Q)\cap S=\{w_1\}$, for otherwise $R-S$ has $6$ vertices and at least $6$ edges.
 This finishes the proof of the $|V(Q)\cap S|=1$ case. 
 From now on we assume $|V(Q)\cap S|\ge 2$.
 
 By Lemma~\ref{lem:common} we can assume that
 $|V(Q)\cap S|= 2$, there is no $(u,v)$-path in $Q-S$ and $V(Q)\cap S \subseteq V(P)$.
 In particular $V(Q)\cap S$ satisfies none of (1)--(4).
 For Rule~\ref{r:del-w1} we are done.
 For the remaining rules we will show that $G'$ has a feedback vertex $S'$ set of size at most $k'$.

 Consider Rule~\ref{r:n(w2)=2}.
 If $|N(w_2)\cap V(P)|=0$, then by Rule~\ref{r:deg2}, $V(P)\subseteq N(w_1)$ and Rule~\ref{r:gamma} applies, a contradiction. 
 If $|N(w_2)\cap V(P)|=1$, then $S_1=(S\setminus V(Q))\cup\{w_1\}\cup(N(w_2)\cap V(P))$ is another feedback vertex set of size at most $k$ in $G$ and we can apply Lemma~\ref{lem:common}.
 Hence $|N(w_2)\cap V(P)| = 2$.
 We claim that in $Q-S$ vertex $w_2$ is reachable from $u$ or $v$.
 We consider cases depending on the distance $d$ in graph $P$ between the two vertices of $N(w_2)\cap V(P)$.
 
 If $d=1$, then $N(w_2)\cap V(P) = \{x_1,x_2\}$ or $N(w_2)\cap V(P) = \{x_2,x_3\}$, since in other cases Rule~\ref{r:deg2} or Rule~\ref{r:gamma} applies.
 By symmetry we can assume $N(w_2)\cap V(P) = \{x_1,x_2\}$ and $\{u,x_3,v\}\subseteq N(w_1)\cap V(P)$.
 It follows that there are at most four possible feedback vertex sets of $Q-S$, namely $\{x_1,x_3\}$, $\{x_2,x_3\}$, $\{x_1,v\}$, and $\{x_2,v\}$.
 The prior two cases are excluded because then $Q-S$ contains a $(u,v)$-path, and in the latter two cases $w_2$ is reachable from $u$ in $Q-S$ as required.
 
 If $d=2$, then $N(w_2)\cap V(P)$ equals $\{u,x_2\}$, $\{x_1,x_3\}$ or $\{x_2,v\}$, where the first and the last case are symmetric, so we skip the analysis of the first one.
 When $N(w_2)\cap V(P)=\{x_2,v\}$, then $\{u,x_1,x_3\}\subseteq N(w_1)\cap V(P)$ by Rule~\ref{r:deg2}.
 It follows that there are at most five possible feedback vertex sets of $Q-S$, namely $\{u,x_2\}$, $\{u,x_3\}$, $\{x_1,x_2\}$, $\{x_1,x_3\}$ and $\{x_1,v\}$. 
 In the first, second, third and fourth case $w_2$ is reachable from $v$ in $Q-S$.
 In the last case $w_2$ is reachable from $u$ in $Q-S$.
 When $N(w_2)\cap V(P)=\{x_1,x_3\}$, then $\{u,x_2,v\}\subseteq N(w_1)\cap V(P)$ by Rule~\ref{r:deg2}.
 It follows that there are at most six possible feedback vertex sets of $Q-S$, namely $\{u,x_2\}$, $\{u,x_3\}$, $\{x_1,x_2\}$, $\{x_1,x_3\}$, $\{x_1,v\}$ and $\{x_2,v\}$. 
 In the first and second case $w_2$ is reachable from $v$ in $Q-S$.
 In the third and fourth case there is a $(u,v)$-path.
 In the fifth and sixth case $w_2$ is reachable from $u$ in $Q-S$.
 
 If $d=3$, then $N(w_2)\cap V(P)$ equals $\{u,x_3\}$ or $\{x_1,v\}$. 
 By symmetry assume the former.
 Then $\{x_1,x_2,v\} \subseteq N(w_1)\cap V(P)$ by Rule~\ref{r:deg2}.
 It follows that there are at most six possible feedback vertex sets of $Q-S$, namely $\{u,x_2\}$, $\{x_1,x_2\}$, $\{x_1,x_3\}$, $\{x_1,v\}$, $\{x_2,x_3\}$ and $\{x_2,v\}$. 
 In the first case $w_2$ is reachable from $v$ in $Q-S$.
 In the second case there is a $(u,v)$-path.
 In the remaining cases $w_2$ is reachable from $u$ in $Q-S$.

 If $d=4$ then Rule~\ref{r:deg2} or Rule~\ref{r:gamma} applies, a contradiction. 
 
 We have thus shown that in $Q-S$, the vertex $w_2$ is reachable from $u$ or $v$. 
 By symmetry assume the former.
 Let $z$ be the vertex of $N(w_2)\cap V(P)$ which is closer to $v$ on $P$.
 Note that neither $w_2$ nor $u$ is reachable from $v$ in $Q-\{w_1,z\}$.
 Hence, the equivalence classes of $R_{Q-\{w_1,z\},\{w_1,w_2,u,v\}}$ are $\{v\}$, $\{w_1\}$ and a partition of $\{u,w_2\}$.
 Then, by Lemma~\ref{lem:fvs-transfer} (applied to $A=B=G$), $S_1=S\setminus V(Q)\cup\{w_1,z\}$ is another feedback vertex set  of size at most $k$ in $G$ and we can apply Lemma~\ref{lem:common}.
  This finishes the proof of correctness of Rule~\ref{r:n(w2)=2}.

 Now consider Rule~\ref{r:gadget:u}.
 Since $S\cap V(Q)$ satisfies none of (1)--(4), we infer that (5) applies to $S\cap V(Q)$.
 Then we pick $S'=S\setminus V(Q)\cup \{y,v\}$.
 Note that the equivalence classes of both $R_{Q'-S',\delta(Q)}$ and $R_{Q-S,\delta(Q)}$ are $\{v\}$, $\{u,w_1,w_2\}$.
 Then, by Lemma~\ref{lem:fvs-transfer} $S'$ is a feedback vertex set in $G'$.
 
 Now consider Rule~\ref{r:gadget:uv}.
 Since $S\cap V(Q)$ satisfies none of (1)--(4), we infer that (5) or (6) applies to $S\cap V(Q)$.
 Then we pick $S'=S\setminus V(Q)\cup \{y,v\}$ or $S'=S\setminus V(Q)\cup \{u,y\}$, respectively.
 Note that the equivalence classes of both $R_{Q'-S',\delta(Q)}$ and $R_{Q-S,\delta(Q)}$ are either $\{v\}$, $\{u,w_1,w_2\}$ or $\{u\}$, $\{v,w_1,w_2\}$.
 Then, by Lemma~\ref{lem:fvs-transfer} $S'$ is a feedback vertex set in $G'$.
 
 Finally consider Rule~\ref{r:gadget:4}.
 Since $S\cap V(Q)$ satisfies none of (1)--(4), we infer that (5) or (7) applies to $S\cap V(Q)$.
 Then we pick $S'=S\setminus V(Q)\cup \{y_1,v\}$ or $S'=S\setminus V(Q)\cup \{y_1,y_2\}$, respectively.
 Note that the equivalence classes of both $R_{Q'-S',\delta(Q)}$ and $R_{Q-S,\delta(Q)}$ are either $\{v\}$, $\{u,w_1,w_2\}$ or $\{u,w_2\}$, $\{v,w_1\}$.
 Then, by Lemma~\ref{lem:fvs-transfer}, $S'$ is a feedback vertex set in $G'$.
\end{proof}

This finishes the proof of Lemma~\ref{lem:indpath5}.

\section{Running time} 
\label{sec:time}

It is easy to see that each of our reduction rules can be detected and performed in $O(n)$ time in such a way that loops and triple edges are not introduced.
Since every rule except for Rule~\ref{r:loop} and Rule~\ref{r:triple} decreases the number of vertices, the total time needed for detecting and performing them is $O(n^2)$.  
In what follows we will show that it can be improved to $O(n)$ expected time. 

We assume that the graph is stored using adjacency lists. Additionally, we use four data structures:

\begin{itemize}
 \item A dictionary $D_1$ implemented as a hash table storing all pairs of adjacent vertices. 
 The hash table stores the corresponding two adjacency list elements for each such pair.
 \item A dictionary $D_2$ implemented as a hash table storing all pairs of vertices $(x,y)$ for which
 the set 
 \[S_{x,y} = \{z\in V\ :\ x,y\in N(z) \text{ and there is at most one edge $zu$ such that $u\not\in\{x,y\}$}\}\]
 is nonempty. The hash table stores the set $S_{x,y}$ for each such pair.
 \item A queue $Q_{3+}$ storing all pairs $(x,y)$ such that $|S_{x,y}|\ge 3$.
 \item A queue $Q_s$ storing vertices with at most four neighbors (not necessarily all). 
\end{itemize}

Once we have the dictionary $D_1$ answering adjacency queries in constant expected time, it is easy to detect and apply Rule~\ref{r:loop} and Rule~\ref{r:triple} immediately after a graph modification. The total expected time needed for that is bounded by total time of graph modifications.
Hence, in what follows we exclude from our considerations Rule~\ref{r:loop} and Rule~\ref{r:triple} and we can assume that application of every rule decreases the number of vertices.

The major challenge in implementing the kernelization algorithm efficiently is {\em detecting} that a rule applies.
In other words we have to find the particular subgraphs described in the rules, which we call {\em configurations}, efficiently.
If a configuration $C$ appears in $G$ then there is an injective homomorphism $h_C:V(C)\rightarrow V(G)$.
Vertices of every configuration are partitioned into two categories: black and white vertices, defined in the figures.
Let $B_C$ and $W_C$ denote the corresponding sets of vertices of $C$.

\begin{lemma}
\label{lemma:detect}
 Assume $Q_{3+}$ is empty.
 Then for every vertex $v$ of $G$ with at most four neighbors one can check in $O(1)$ time whether there is a configuration $C$ in $G$ such that $h_C^{-1}(v)$ is a black vertex.
\end{lemma}

\begin{proof}
 Since the number of configurations used in the algorithm is bounded, we can consider a fixed configuration $C$.
 We claim that there are only $O(1)$ candidate homomorphisms to check.
 Note that in each configuration, every white vertex has a black neighbor.
 It follows that it suffices to show that there are $O(1)$ mappings of all the black vertices of $C$ to $V(G)$ that can extend to a homomorphism $h_C$, since the configurations are bounded and black vertices have bounded degree.
 This claim is immediate for configurations where black vertices induce a connected subgraph.
 The remaining configurations are those from Rule~\ref{r:3deg3}, Rule~\ref{r:digon:1+2:1}, and Lemma~\ref{lem:digon:1+2:2}.
 Since $Q_{3+}$ is empty, Rule~\ref{r:3deg3} does not apply.
 In the configurations from the remaining two rules black vertices induce exactly two connected components.
 Moreover, for both these configurations, once we fix a mapping $h$ of black vertices in the connected component $Q$ of $C[B_C]$ such that $h^{-1}(v)\in V(Q)$, then there are vertices $a\in V(Q)$ and $b\in B_C\setminus V(Q)$ such that for some pair of white vertices $x,y$ we have $x,y\in N_C(a)$ and $x,y\in N_C(b)$.
 Note also that in both cases there is at most one edge $h_C(b)u$ such that $u\not\in\{h_C(x),h_C(y)\}$.
 But it means that the homomorphic image of $b$ must belong to $S_{h_C(x),h_C(y)}$ and since $Q_{3+}=\emptyset$, there are only $O(1)$ candidates for it.
 Then there are only $O(1)$ candidates for other black vertices in the connected component of $C[B_C]$ containing $b$.
 The claim follows.
\end{proof}

Now we can describe our algorithm.

\begin{algorithm}[H]
\label{alg:kernel}
\caption{\textsc{Kernelize}$(G,k)$}
  Initialize $D_1$, $D_2$, $Q_{3+}$ according to their definitions\;
  Initialize $Q_s$ with all vertices of $G$ with at most four neighbors\;
  \While{$Q_{3+}\cup Q_s\ne \emptyset$}{
    \eIf{$Q_{3+}\ne \emptyset$}{%
      Remove an element from $Q_{3+}$ and apply Rule~\ref{r:3deg3}\;%
    }{%
      Remove a vertex from $Q_{s}$ and apply Lemma~\ref{lemma:detect}\;%
      If a configuration is found, apply the corresponding rule\;
      }
    }
\end{algorithm}

The correctness of Algorithm~\ref{alg:kernel} follows from the following invariants.

\smallskip

\noindent
{\bf Invariant 1} The information stored in $D_1$, $D_2$, $Q_{3+}$ is up to date.
\smallskip

\noindent
{\bf Invariant 2} If a configuration $C$ appears in $G$, then there is a vertex $z\in Q_s$ such that for some black vertex $v\in V(C)$ we have $h_C(v)=z$.
\smallskip

Clearly, both invariants hold before before the while loop (Invariant 2 holds since all the black vertices have at most four neighbors).
Moreover, after every modification of $G$ resulting in an application of a rule we update the data structures so that both invariants hold, as follows.

\begin{itemize}
 \item Whenever an adjacency list changes, we update $D_1$.
 \item If the set of incident edges of a vertex $v$ with at most four neighbors changes, then we add $v$ to $Q_s$, and for every pair $(x,y)$ of its neighbors we add the pair to $D_2$ if needed; if $|S_{x,y}|$ grows to three, we add the pair to $Q_{3+}$.
 \item If an edge $xy$ is added to $G$, we add all elements of $S_{x,y}$ to $Q_s$. 
       Note that then $Q_{3+} = \emptyset$, since in Rule~\ref{r:3deg3} no edges are added.
       Hence $|S_{x,y}|\le 3$.
 \end{itemize}

It is easy to check that the updates described above guarantee that both invariants are satisfied (note that the last item above is needed only to guarantee Invariant 1 for configurations of Rule~\ref{r:gamma}, since this is the only configuration where two white vertices must be adjacent).
We are left with the time complexity analysis.
Observe that the expected time of the algorithm is bounded by a function which is linear in the total number of insertions to $Q_{3+}$ and $Q_s$.
(A lookup, insert or delete opertion in a hash table works in $O(1)$ expected time and this is the only source of randomness in the running time.)
The number of insertions to $Q_{3+}$ and $Q_s$ is linear in the size of the input graph added to the number of applications of the rules.
Since each application decreases the number of vertices, there are at most $n$ of them. 
Hence the total expected time is bounded by $O(n)$.
Note that we can turn it to $O(n\log n)$ deterministic time by replacing hash tables by balanced binary search trees.

\section{Concluding remarks and further research}
\label{sec:conclude}

We have shown a kernel of $13k$ vertices for \probPlFVS. 
Our main contribution was applying the region decomposition technique in a new way.
It would be interesting to see more applications of the region decomposition technique in problems in which it was not used before.

An obvious open problem is improving the kernel size even further. 
In particular, it would be nice to break the psychological barrier or {\em single digit kernel}, i.e., to get a $9k$-kernel.
We suppose that if this is possible, it would require finding a number of new, very specialized reduction rules.

\bibliographystyle{abbrv}
\bibliography{fvs}

\end{document}